\documentclass{article}
\usepackage{fullpage}
\usepackage[utf8]{inputenc}
\usepackage{latexsym, amssymb, dsfont}
\usepackage{amsmath, amsthm}
\usepackage{enumerate}
\usepackage{tikz}

\usetikzlibrary{circuits.logic.US}
\pgfdeclarelayer{background}
\pgfsetlayers{background,main}
\usetikzlibrary{arrows,shapes,decorations,decorations.pathreplacing,decorations.pathmorphing}

\newtheorem{thm}{Theorem}
\newtheorem{lem}{Lemma}
\newtheorem{cor}{Corollary}
\theoremstyle{definition}

\newtheorem{definition}{Definition}

\gdef\dash---{\thinspace---\hskip.16667em\relax} 
\gdef\op|{\,|\;}

\newcommand{\figref}[1]{Fig.~\ref{#1}}

\newcommand{\bfno}[1]{\noindent{\bf #1}}

\newcommand{\Time}{\mathcal{T}}

\newcommand{\Bool}{\mathbb{B}}

\newcommand{\li}[1]{\left({#1}\right)}

\newcommand{\IR}{\mathds{R}}

\newcommand{\bool}{\Bool}

\newcommand{\myvspace}[1]{}

\usetikzlibrary{petri}
\tikzstyle{binary place}=[place,circle, double]
\tikzstyle{node}=[circle,draw=black,thick,minimum size=9mm]
\tikzstyle{dest}=[circle,draw=black!50,fill=black!20,thick,minimum
  size=9mm]
\tikzstyle{post}=[->,thick]
\tikzstyle{pre}=[<-,thick]
\tikzstyle{every transition}=[fill,minimum width=1cm,minimum height=2mm]
\tikzstyle{Atransition}=[transition,fill,minimum width=1cm,minimum height=2mm]
\tikzstyle{Otransition}=[transition,fill=white,minimum width=1cm,minimum height=2mm]
\tikzstyle{THtransition}=[transition,fill=white,minimum width=4mm,minimum height=1cm]

\newcommand{\emdot}[1]{\fill[black] #1 circle (0.5mm) node[left] {};%
  \fill[white] #1 circle (0.2mm) node[left] {};%
}

\newcommand{\fidot}[1]{\fill[black] #1 circle (0.5mm) node[left] {};}

\newcommand{\feline}[3]{
  \path [draw,line width=0.5mm,-] (#1,#3) -- (#2,#3);%
  \fidot{(#1,#3)}%
  \emdot{(#2,#3)}%
}

\newcommand{\neline}[3]{
  \path [draw,line width=0.5mm,-] (#1,#3) -- (#2,#3);%
  \emdot{(#2,#3)}%
}

\newcommand{\fnline}[3]{
  \path [draw,line width=0.5mm,-] (#1,#3) -- (#2,#3);%
  \fidot{(#1,#3)}%
}

\newcommand{\tcross}[2]{%
  \draw plot[mark=x, mark size=2mm] coordinates{#1} node[above=0mm]{#2};%
}

\newcommand{\tcrossAR}[2]{%
  \draw plot[mark=x, mark size=2mm] coordinates{#1} node[above right=0mm]{#2};%
}

\tikzstyle{Tdelay} = [draw, rectangle, rounded corners,
    minimum height=4mm, minimum width=20mm]

\tikzstyle{Tfunction} = [draw, rectangle,
    minimum height=4mm, minimum width=4mm]

\tikzstyle{Tsignal} = [draw,fill=black,circle, size=1mm]

\tikzstyle{ra} = [draw,thick,double,double distance=1.0pt,->]
\tikzstyle{r} = [draw,->,line width=0.5pt]

\title{Unfaithful Glitch Propagation in Existing Binary Circuit Models}

\author{Matthias Függer\textsuperscript{1} \and Thomas Nowak\textsuperscript{2}
\and Ulrich Schmid\textsuperscript{1}}

\date{\textsuperscript{1} ECS Group, TU Wien, Austria\\\textsuperscript{2} LIX,
\'Ecole
polytechnique, France}

\begin{document}
\maketitle

\begin{abstract}
We show that no existing continuous-time, binary value-domain  model
     for digital circuits is able to correctly capture glitch
     propagation.
Prominent examples of such models are based on pure delay channels
     (P), inertial delay channels (I), or the elaborate PID channels
     proposed  by Bellido-D\'{i}az et al.
We accomplish our goal by considering the solvability/non-solvability
     border of a simple problem called Short-Pulse Filtration (SPF),
     which is closely related to arbitration and synchronization.
On one hand, we prove that SPF is solvable in bounded time in any such
     model that provides channels with non-constant delay, like I and
     PID.
This is in opposition to the impossibility of solving bounded SPF in
     real (physical) circuit models.
On the other hand, for binary circuit models with constant-delay
     channels, we prove that SPF cannot be solved even in unbounded
     time; again in opposition to physical circuit models.
Consequently, indeed none of the binary value-domain models proposed
     so far (and that we are aware of) faithfully captures glitch
     propagation of real circuits.
We finally show that these modeling mismatches do not hold for the
     weaker eventual SPF problem.
\end{abstract}

\section{Introduction}

Binary value-domain models that allow to model glitch
propagation have always been of interest, especially in
asynchronous design \cite{unger71}: Pure delay channels and inertial delay
channels, which propagate input pulses with some constant delay
only when they exceed some minimal duration, are still the
basis of most digital timing analysis approaches and tools. The
tremendous advances in digital circuit technology, in 
particular increased speeds and reduced voltage swings,
raised concerns about the accuracy of these models~\cite{BDJCAVH00}.
For example, neither pure nor inertial delay models can express 
the well-known phenomenon of propagating glitches that decay
from stage to stage, which is particularly important for
analyzing high-frequency pulse trains or oscillatory metastability~\cite{marino81}.

At the same time, the steadily increasing complexity of contemporary
digital circuits fuels the need for fast digital timing analysis 
techniques: Although accurate Spice models, which facilitate very
precise analog-level simulations, are usually available
for those circuits, the achievable simulation times are 
prohibitive. Refined digital timing analysis models like the PID model proposed by
Bellido-D\'{\i}az et~al.\ \cite{BDJCAVH00}, which is both fast and more accurate,
are hence very important from a practical perspective \cite{BJV06}.

The interest in binary models that faithfully model
glitch propagation and even metastability has also been
stimulated recently by the increasing importance of incorporating fault-tolerance
in circuit design \cite{Con03}: Reduced
voltage swings and smaller critical charges make circuits more susceptible to
particle hits, crosstalk, and electromagnetic
interference~\cite{GEBC06,MA01}. Since single-event transients, caused
by an ionized particle hitting a reverse-biased transistor, just manifest
themselves as short glitches, accurate propagation models are important
for assessing soft error rates, in particular, for asynchronous circuits.
After all, if
system-level fault-tolerance techniques like triple modular redundancy are
used for transparently masking value failures, the only remaining issue
are timing failures, among which glitches are the most problematic ones.

For example, the DARTS Byzantine fault-tolerant distributed clock 
generation approach~\cite{FS12:DC} employs standard 
asynchronous circuit components, like micropipelines \cite{Suth89},
which store clock ticks received from other nodes; a new clock tick is
generated when sufficiently many micropipelines are non-empty. Clearly,
since any ``wait-for-all'' mechanism may deadlock in the presence of faulty 
components, handshaking had to be replaced by threshold logic in conjunction with some bounded delay 
assumptions. This way, DARTS can tolerate arbitrary behavior of Byzantine faulty nodes, except
for the generation of pulses with a duration that drive the Muller C-elements of a
pipeline into metastability. Analyzing the propagation of such pulses
along a pipeline is thus important in order to assess the achievable resilience
against such threats \cite{FFS09:ASYNC09}.
The situation is even worse in case of self-stabilizing algorithms 
\cite{Dol00}, which must be able to recover from an arbitrary initial/error state: Neither
handshaking nor any bounded delay condition can be resorted to during
stabilization in an algorithm like the one presented by Dolev et al.~\cite{DFLS11:sss}. Consequently,
glitches and the possibility of metastability cannot be avoided.

As a consequence, discrete-value circuit models, analysis techniques and 
supporting tools for a fast but nevertheless accurate glitch and metastability 
propagation analysis will be a key issue in the design of future VLSI circuits. 
In this paper, we rigorously prove that none of the existing binary-value
candidate models proposed in the past captures glitch propagation adequately.
Searching for alternative models is hence an important challenge for future
research on asynchronous circuits.

\medskip

\bfno{Detailed contributions.}
In Section~\ref{sec:physical}, we define the Short-Pulse Filtration (SPF)
problem in the physical
circuit model of Marino and recall the behavior of physical circuits with
respect to SPF.
That is, we show that unbounded SPF is solvable with physical circuits while
bounded SPF is not.
The SPF problem is closely related to glitch  propagation, as it is
     essentially the problem of building a one-shot  inertial channel.

In Section~\ref{sec:model}, we present a generic binary value-domain
     model for digital clocked and clockless circuits, and
     introduce the SPF problem.
Our generic model comprises zero-time logical gates interconnected by
     channels that encapsulate model-specific propagation delays and
     related decay effects.
Non-zero time logical gates can be expressed by appending channels
     with delay at the gate's inputs and outputs.
The simplest channel is a pure delay channel, which propagates its
     input signal with a fixed delay and without any decay, i.e., a
     pulse has the same duration at the channel's input and output.

In Section~\ref{sec:P}, we prove that even unbounded SPF is
     unsolvable when only pure, i.e., constant-delay channels are
     available.
This is in contrast with the solvability result with physical circuits of
Section~\ref{sec:physical}.

In Section~\ref{sec:single}, we turn our attention to a generalization of
     constant-delay channels, termed \emph{bounded single-history channels},
     which are FIFO channels with a generalized delay function that
     also takes into consideration the last output transition.
We distinguish between {\em forgetful\/} and {\em non-forgetful\/}
     single-history channels, depending on their behavior when a pulse
     disappears at the output due to decay effects.
All existing binary models we are aware of can be expressed as
     single-history channels with specific delay functions: A pure
     delay channel (P) as either a forgetful or non-forgetful
     single-history channel, a classical inertial delay channel (I) as
     a forgetful single-history channel, and the channel model
     proposed by Bellido-D\'{i}az et al.~\cite{BDJCAVH00} (PID), 
     which additionally has a decay component, as a non-forgetful
     single-history channel.

In Section~\ref{sec:PD}, we prove that bounded SPF is solvable if just
     a single forgetful or non-forgetful single-history channel with
     non-constant delay is available.
However, this is again in contradiction with the result of Section~\ref{sec:physical}
showing impossibility of bounded SPF with physical circuits.

In Section~\ref{sec:eSPF}, we prove that weakening SPF to eventual SPF
     fails to witness the above modeling mismatch: Eventual SPF can by
     solved both with single-history and physical channels.

\figref{fig:results} summarizes our (im)possibility results.

\def\mysep{0.6} 
\begin{figure}[bht]
\centering
\begin{tikzpicture}

\draw (0,{3*\mysep}) -- ++(6,0);
\draw (0,{2*\mysep}) -- ++(6,0);
\draw (0,{\mysep}) -- ++(6,0);
\draw (0,0) -- ++(6,0);

\draw (6.0,0) -- ++(0,{3*\mysep});
\draw (4.5,0) -- ++(0,{3*\mysep});
\draw (3.0,0) -- ++(0,{3*\mysep});
\draw (1.5,0) -- ++(0,{3*\mysep});
\draw (0.0,0) -- ++(0,{3*\mysep});

\node[left] at (-0.3,{2.5*\mysep}) {bounded SPF};
\node[left] at (-0.3,{1.5*\mysep}) {SPF};
\node[left] at (-0.3,{0.5*\mysep}) {eventual SPF};

\node at ({0+0.75},{-0.4}) {\small constant};
\node at ({1.5+0.75},{-0.44}) {\small forgetful};
\node at ({2*1.5+0.75},{-0.4}) {\small non-};
   \node at ({2*1.5+0.75},{-0.7}) {\small forgetful};
\node at ({3*1.5+0.75},{-0.4}) {\small physical};

\node at (0.75,{2.5*\mysep}) (a) {X};
\node at (0.75,{1.5*\mysep}) (b) {X};
\node at (0.75,{0.5*\mysep}) {\Large $\checkmark$};

\draw[->] ($(b.center)+(0,0.2)$) -- ($(a.center)+(0,-0.2)$);

\node at ({1.5+0.75},{2.5*\mysep}) (a2) {\Large $\checkmark$};
\node at ({1.5+0.75},{1.5*\mysep}) (b2) {\Large $\checkmark$};
\node at ({1.5+0.75},{0.5*\mysep}) (c2) {\Large $\checkmark$};

\draw[->] ($(a2.center)+(0,-0.2)$) -- ($(b2.center)+(0,0.2)$);
\draw[->] ($(b2.center)+(0,-0.2)$) -- ($(c2.center)+(0,0.2)$);

\node at ({2*1.5+0.75},{2.5*\mysep}) (a3) {\Large $\checkmark$};
\node at ({2*1.5+0.75},{1.5*\mysep}) (b3) {\Large $\checkmark$};
\node at ({2*1.5+0.75},{0.5*\mysep}) (c3) {\Large $\checkmark$};

\draw[->] ($(a3.center)+(0,-0.2)$) -- ($(b3.center)+(0,0.2)$);
\draw[->] ($(b3.center)+(0,-0.2)$) -- ($(c3.center)+(0,0.2)$);

\node at ({3*1.5+0.75},{2.5*\mysep}) (a4) {X};
\node at ({3*1.5+0.75},{1.5*\mysep}) (b4) {\Large $\checkmark$};
\node at ({3*1.5+0.75},{0.5*\mysep}) (c4) {\Large $\checkmark$};

\draw[->] ($(b4.center)+(0,-0.2)$) -- ($(c4.center)+(0,0.2)$);

\end{tikzpicture}
\caption{Possibility~($\checkmark$) and Impossibility~(X)
     Results for constant, non-constant forgetful, non-const.\
     non-forgetful, and physical physical channels.
Arrows mark implications.} 
\label{fig:results}
\end{figure}

\medskip

\bfno{Related Work.}
Unger~\cite{unger71} proposed a general technique for deriving
     asynchronous sequential switching circuits that can cope with
     unrelated input signals. It assumes signals to be binary valued, 
and requires the availability of combinational circuit elements, as well as
     pure and inertial delay channels.

Bellido-D{\'\i}az et~al.\ \cite{BDJCAVH00} proposed the PID model, and justified
its appropriateness both analytically and by comparing the model predictions
against Spice simulation results. The results confirm very good accuracy 
even for such challenging scenarios as long chains of gates and ring 
oscillators.

Marino~\cite{marino77} showed that the problem of building a
     synchronizer can be reduced to the problem of building an
     inertial delay channel.
The reduction circuit only makes use of combinational gates and pure
     delay channels in addition to inertial delay channels.
Marino further shows, in a continuous value signal model, that for a
     set of standard designs of inertial delay channels, input pulses
     exist that produce outputs violating the requirements of inertial
     delay channels.
Barros and Johnson~\cite{BJ83} extended this work, by showing the
     equivalence of arbiter, synchronizer, latch, and inertial delay
     channels.

Marino~\cite{marino81} developed a general theory of metastable
     operation, and provided impossibility proofs for
     metastability-free synchronizers  and arbiter circuits for
     several continuous-value circuit models.
Branicky~\cite{Bra95} proved the impossibility of  time-unbounded
     deterministic and time-invariant arbiters modeled as ordinary
     differential equations.
Mendler and Stroup~\cite{ms93}  considered the same problem in the
     context of continuous~automata.

Brzozowski and Ebergen~\cite{BE92} formally proved that, in a model
that uses only binary values, it is
impossible to implement Muller C-Elements (among other basic
     state-holding components used in (quasi) delay-insensitive
     designs) using only zero-time logical gates interconnected
     by wires without timing restrictions.


\section{Short-Pulse Filtration in Physical Systems}\label{sec:physical}

In this section, we will introduce the SPF problem in the model of 
Marino~\cite{marino81} and use the classic results obtained
for bistable elements to determine the 
solvability/impossibility border of the SPF problem for real (physical) 
circuits.

The model of Marino considers circuits which process
signals with both continuous value domain and continuous time domain.
Accordingly, we
assume (normalized) signal voltages to be within $[0,1]$, and denote
by $L_0  = [0,l_0]$ resp.\ $L_1 = [l_1,1]$, with $0 < l_0 < l_1 <
     1$, the signal ranges that are interpreted as logical~$0$
     resp.\  logical~$1$ by a circuit.

A physical circuit with a single input and a single output
{\em solves Short-Pulse Filtration (SPF)\/}, if it
fulfills the
     following requirements:
\begin{enumerate}
\item[(i)] If the input signal is constantly logical~$0$, then so is the
     output signal.
\item[(ii)]
There exists an input signal such that the
     output signal attains logical~$1$ at some point in time.
\item[(iii)] There exists some fixed~$\varepsilon>0$ such that, 
if the output signal
     is not interpreted as logical~$1$ at two points in time~$t$ and $t'$ with
     $t'-t < \varepsilon$, then it is not logical~$1$ at any time in
     between~$t$ and~$t'$. Informally, this condition prohibits output signals
that may be interpreted as pulses (see Section~\ref{sec:SPF})
with a duration less than~$\varepsilon$.
\end{enumerate}
A physical circuit {\em solves bounded SPF\/} if additionally:
\begin{enumerate}
\item[(iv)] There exists a time~$T$ such that, if the input signal
     switches to logical~$1$ by time~$t$, then the
     output signal value is either logical~$0$ or logical~$1$ at
     time~$t+T$ and remains logical~$0$ respectively logical~$1$ thereafter.
\end{enumerate}

We will next argue why there is no physical circuit that solves
     bounded SPF, but that there are physical circuits solving unbounded
     SPF.

\subsection{Impossibility of Bounded SPF}

The proof is by reduction to the non-existence of a physical
bistable storage element that stabilizes within bounded
time in the model of Marino.
A {\em single-input bistable element\/} is a physical circuit
     with a single input and a single output that
     fulfills properties~(i) and~(ii) of SPF as well as:  
\begin{enumerate}
\item[(iii')] If the output is logical~$1$ at some time~$t$, it also
     remains logical~$1$ at all times larger than~$t$.
\end{enumerate}
For a {\em single-input bistable element stabilizing within bounded
     time}, additionally~(iv) has to hold. 

The following Corollary~\ref{cor:bistable}, which proves
the non-existence of a single-input bistable element that
stabilizes within bounded time,
follows immediately from Theorem~3 in~\cite{marino81}.


\begin{cor}\label{cor:bistable}
There is no single-input bistable element stabilizing within bounded time.
\end{cor}

Now assume, for the sake of a contradiction, that there existed 
a physical circuit solving bounded SPF and consider the circuit
shown in \figref{fig:reduction}, with the NOR's initial
     output equal to~$1$ and the inverter's initial output equal
     to~$0$ at time $t=0$.

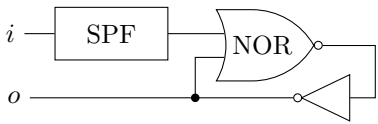
\begin{figure}[hbt]
\centering
\begin{tikzpicture}[circuit logic US, scale=1]
\matrix[column sep=0mm,row sep=-1mm,every node/.style={transform shape}]
{
              & & \node [nor gate,small circuit symbols] (nor) {\,\,NOR}; & \\
\node (k) {}; & & \node [not gate,small circuit symbols,rotate=180,xshift=-10mm] (not) {}; & \\
};

\draw (nor.input 1)
        -- ++(left:1.5) node[draw,rectangle,minimum width=1.5cm,minimum height=0.7cm,fill=white] (spf) {SPF};
\draw (spf.west)
        -- ++(left:0.4) node[left] (i) {$i$};
\draw (nor.output) 
	-- ++(right:0.7) 
	-- ++(down:0.1)
        |- (not.input);
\draw (not.output)
        -- ++(left:1.3) node[circle,inner sep=1pt,fill=black,draw] (huhu) {}
	-- ++(up:0.1)
	|- (nor.input 2)
	;
\draw (huhu)
        -- ++(left:2.2)
	node[left] {$o$}
	;
\end{tikzpicture}
\myvspace{-0.2cm}
\caption{Building a 
bistable storage element
  from a circuit solving SPF}
\label{fig:reduction}\myvspace{-0.2cm}
\end{figure}

It is not difficult to prove that this circuit implements 
a single-input bistable element stabilizing within bounded time:
In case the input signal~$i$ is always logical~$0$, the SPF's output signal
     will always be logical~$0$ due to property~(i) of the SPF.
Thus the circuit shown in \figref{fig:reduction} 
will always drive a logical~$0$ at its
     output, which confirms property~(i) for the bistable element.

Now let~$u$ be an input pulse that makes the SPF circuit produce 
a logical~$1$ at its output. 
Letting~$t'$ be the first time the SPF circuit drives a logical~$1$
at its output, its output must remain logical~$1$ within 
$[t',t'+\varepsilon]$ for some $\varepsilon > 0$ due to 
property~(iii) of the SPF. Assuming that the signal propagation
delay of the NOR gate and the inverter is short enough
     for the inverter's output to reach a logical~$1$ before
     time~$t'+\varepsilon$, the NOR gate will subsequently
drive a logical~$0$ on its output forever, irrespective 
of the output of the SPF circuit. The circuit's output signal~$o$
will hence continuously remain logical~$1$ once it
     switched to logical~$1$, which also confirms
properties~(ii) and~(iii') of the bistable element. 

Due to the use of a circuit solving bounded SPF in the compound
     circuit, we further obtain that there exists some $T>0$ such that,
     for any input pulse~$u'$ that switches to logical~$1$ by
     time~$t$, the circuit shown in \figref{fig:reduction} 
produces a logical~$1$ by
     time~$t+T$, a contradiction to the non-existence of a 
single-input  bistable element stabilizing in bounded time.
We hence obtain:

\begin{thm}\label{thm:bSPFphysimposs}
No physical circuit solves bounded SPF.
\end{thm}

\subsection{Possibility of Unbounded SPF}

To show the existence of a circuit solving unbounded SPF, we make use
     of a circuit known as a metastability filter (see, e.g.,
     \cite[p.~40]{Kin07}). According to Marino~\cite{marino81},
pulses of arbitrary length may drive the internal state of 
every storage loop (including the one shown in \figref{fig:reduction})
into a metastable region for an unbounded time. A circuit may hence
produce an output signal within some region of metastable output
values $[v_M^-,v_M^+] \subset [0,1]$ during an unbounded time,
where the values $v_M^-,$ and $v_M^+$ depend on technology 
parameters. However, since it is possible to compute safe bounds
$V_M^-,$ and $V_M^+$ such that $[v_M^-,v_M^+] \subset [V_M^-,V_M^+]
\subset [0,1]$, a continuously valid
output signal can be produced by means of a subsequent high-threshold buffer:
By connecting the output~$o$ of \figref{fig:reduction}, ignoring the SPF block,
to the input of a (high-threshold) buffer, which maps 
input signal values within
     $[0,B_M^-]$ to output signal values that are logical~$0$, and
     input values within $[B_M^+,1]$ to output
     values that are logical~$1$, where $V_M^+<B_M^-$,
we obtain a physical circuit that solves (unbounded)~SPF.
Hence:

\begin{thm}\label{thm:SPFphysposs}
There is a physical circuit that solves SPF.
\end{thm}

\section{Binary System Model}\label{sec:model}

\subsection{Signals, Events}

We consider a binary valued signal model with continuous time, i.e.,
     signal values are from~$\Bool=\{0,1\}$ and they evolve over
     time~$\Time=[0,\infty)$.

A {\em signal\/} is a function~$\Time\to\bool$ that
does not change an infinite number of times
     during a finite time interval and that already has its new value at a time
instant of a value transition.\footnote{The requirement that a signal
already has its new value when changing values is merely a convention. On the other
hand, the requirement that it only changes a finite number of times during a
finite time interval is fundamental to our model and, thus, our results.
}

A signal transition is modeled by an {\em event}.
Formally an event is a pair~$e=(t,x)$ in~$(\Time \cup \{-\infty\})\times\bool$.
We call~$t$ the event's {\em time\/} and~$x$ the event's {\em value}.
We use ``virtual
events'' at~$t=-\infty$ to simplify notation when specifying initial values of
channels.
An {\em event list\/} is a (finite or infinite) sequence of events.

To every signal, there corresponds an
     event list $(e_n)=(t_n,x_n)$ with the following properties:  
\begin{enumerate}[S1)]
\item There is always an initial event at time~$-\infty$.
\item The sequence~$(t_n)$ of event times is strictly increasing and discrete.
\item Values are alternating: $x_n\neq x_{n+1}$
\end{enumerate}
Conversely, every such event list corresponds to a unique signal.

\subsection{Channels, Constant-Delay Channels}
\label{sec:channels}

A {\em channel\/}~$c$ is a function mapping an input signal~$s$ to
     an output signal~$c(s)$.

The simplest class of channels is the class of (positive) constant-delay channels.
A {\em constant-delay\/} channel~$c$ with delay parameter~$\delta > 0$ and initial
value~$x\in\Bool$ produces at its output the input signal delayed by~$\delta$, i.e.,
\begin{equation}\label{eq:p:channel}
	c(s)(t) = \begin{cases} x & \text{if }t<\delta\\ s(t-\delta) & \text{if }t\geq \delta \enspace. \end{cases}
\end{equation}

Note that a physical realization of a constant-delay channel with initial
value~$x$ 
requires a multiplexer, which supplies the constant-delay channel's input
with the initial value $x$ during $(-\infty,0)$ and switches to
the actual input $s$ at reset time~$0$.

\subsection{Circuits}

Circuits are obtained by interconnecting a set of input ports and a set
     of output ports, forming the external interface of a circuit, 
and a set of combinational gates via channels. 
We constrain the way components are interconnected in a natural
way, by requiring
     that input ports are attached to one or 
more channel inputs only (C4), and that both output
     ports and gate inputs are attached to just one channel's output (C5, C6);
the latter prevents channel outputs driving against each other.

Formally, a {\em circuit\/} is a tuple~$C=(G,I,O,c,n)$, where  
\begin{enumerate}[C1)]
\item $G$ is a directed graph whose vertex set can be partitioned as $I\cup
O\cup B$.
\item Every vertex~$b$ in~$B$ ({\em (Boolean) gate\/}) is assigned a
     Boolean function $\bool^{d_b}\to\bool$, where~$d_b$ is the in-degree,
     i.e., the number of incoming neighbors, of~$b$.
By a slight abuse of notation, $b$ also denotes the Boolean function
     assigned to~$b$.

\item $c$ is a function that maps every edge $(u,v)$ in $G$ to its
corresponding channel $c_{u,v}$.
\item Every vertex $v\in I$ ({\em input ports\/}) has in-degree~$d_v=0$.
\item Every vertex $v\in O$ ({\em output ports\/}) has in-degree~$d_v=1$.
\item $n$ is a function that maps every vertex $v$ in $G$ to
a linearly ordered subset $n_v=\{v_1,\dots,v_{d_v}\}$ of its in-neighbor
vertices in $G$, i.e., where edge $(v_i,v)$ for $i=1$ up to $v$'s
in-degree $d_v$ is in $G$.
\end{enumerate}
Note that there are also zero-input Boolean gates~$0$ and~$1$ that
represent constant signal values~0 and~1.

\subsection{Executions}\label{subsec:exec}
An execution of circuit~$C$ is an assignment of signals to vertices that
respects the channel functions and Boolean gate functions.

Formally, an {\em execution\/} of circuit~$C$ is a collection of signals~$s_v$
for all vertices~$v$ of~$C$ such that the following properties hold:
If~$i$ is an input port, then there are no restrictions on~$s_i$.
If~$o$ is an output port, then~$s_o = c_{v,o}(s_v)$ where~$v$ is the unique
     incoming neighbor of~$o$ and $c_{v,o}$ the channel representing edge
$(v,o)$.
Let now~$b$ be a Boolean gate with~$d$ incoming
neighbors~$v_1,v_2,\dots,v_d$, ordered according to~$n_b$.
We then apply, for each incoming edge~$(v_k,b)$, the channel~$c_{v_k,b}$
to signal~$s_{v_k}$ and check that
the signal value~$s_b(t)$ is the gate's Boolean combination of
these incoming signals at time~$t$.
That is,
$s_b(t) = b\big(
     c_{v_1,b}(s_{v_1})(t)\,,\,\dots\,,\,c_{v_d,b}(s_{v_d})(t)
     \big)$ for all~$t\in\Time$.

Not all circuits necessarily do have executions.
For example, the circuit comprising a single inverter gate whose output is
     fed back to its input via the ``mirror channel'' $c$ with $c(s)=s$
     for all
     signals~$s$ does not have an execution.
Whenever we introduce a circuit for a possibility result, we will thus make sure that it allows
     for a unique execution once the input signals are fixed.
In case of constant-delay channels, this is always
     the case (see Lemma~\ref{lem:executions}).

\subsection{Short-Pulse Filtration}\label{sec:SPF}

A {\em pulse\/}~$p$ of length~$\Delta>0$ at time~$T$ is a signal of the form
\begin{align}
  p(t) = \begin{cases}
		0 & \text{if } t < T \text{ or } t \geq T+\Delta\\
		1 & \text{if } T\leq t < T+\Delta \enspace.
		\end{cases} 
\end{align}

A signal {\em contains a pulse\/} of length~$\Delta>0$ at time~$T$ if its
event list contains the two consecutive events $(T,1)$ and $(T+\Delta,0)$.

A circuit {\em solves Short-Pulse Filtration (SPF)\/} if it fulfills the following conditions:
\begin{enumerate}[F1)]
\item It has exactly one input port~$i$ and exactly one output port~$o$.
\item For every pulse~$p$, there exists an execution that has~$p$ as the input
	signal (i.e., $s_i=p$). {\em (Well-formedness)}
\item In all executions, if the input signal is constant zero, then so is the output
	signal. {\em (No generation)}
\item There exist a pulse~$p$ such that, in all executions with~$p$ as the
	input signal, the output signal is not the constant zero signal. {\em (Nontriviality)}
\item There exists an~$\varepsilon>0$ such that, in all executions, the output
	signal does not contain a pulse of length less than~$\varepsilon$. {\em (No short pulses)}
\end{enumerate}

A circuit {\em solves bounded SPF\/} if additionally the following condition holds:
\begin{enumerate}
\item[F6)] There exists a $K>0$ such that, in all executions with a
  pulse as the input signal
whose last event is at time~$T$, the output
	signal does not change anymore after time~$T+K$.
	{\em (Bounded stabilization time)}
\end{enumerate}

A circuit {\em solves eventual SPF\/} if conditions (F1)--(F4) and the following condition hold:
\begin{enumerate}
\item[F5e)] There exists an~$\varepsilon>0$ and a $K>0$ such that, in all
executions with a pulse at time~$T$ as the input signal, the output
	signal does not contain a pulse of length less than~$\varepsilon$ after
	time~$T+K$. {\em (Eventually no short pulses)}
\end{enumerate}





\section{Unsolvability of Short-Pulse Filtration with Constant-Delay Channels}\label{sec:P}

In this section, we show that no circuit whose channels are all positive constant-delay
channels solves SPF.
The idea of the proof is to exploit the fact that the value of the
     output signal of the circuit at each time~$t$ only depends on a
     {\em finite\/} number of values of the input signal at times~$t'$
     between~$0$ and~$t$.

Calling each such time~$t'$ a {\em measure point\/} for time~$t$, we
     show that indeed only a finite number of measure points exists
     for time~$t$, i.e., the circuit cannot distinguish two
     different input signals that do not differ in the input signal
     values at the measure points for time~$t$: For both such input
     signals, the output signal must have the same value at time~$t$.
Combining that indistinguishability result with a shifting
     argument of the input signal allows us to construct an
     arbitrary short pulse at the output of the circuit, a
     contradiction to property (F5) of Short-Pulse Filtration.

\subsection{Dependence Graphs}
For each constant-delay circuit with a single input port and a single output
port, we introduce its {\em dependence graph}, which describes the way the
output signals may depend on the input signals.

Let~$C=(G,I,O,c,m)$ be a circuit with constant-delay channels, a single input
port~$i$, and a single output port~$o$.
For every channel~$c_{u,v}$ of~$C$, denote by~$\delta(u,v)$ its delay
parameter~$\delta$ and by~$x(u,v)$ its initial value.
The {\em dependence graph\/}~$DG(t)$ of~$C$ at time~$t$ is a directed
graph with vertices~$(v,\tau)$, where $v$ is a vertex in~$G$ and $\tau$ a time.
It is defined as follows:
\begin{itemize}
\item The pair~$(o,0)$ is a vertex of~$DG(t)$.
\item If~$(v,\tau)$ is a vertex of~$DG(t)$ and~$(u,v)$ is an edge in~$G$
such that $\tau+\delta(u,v)\leq t$, then the pair $\big(u,\tau+\delta(u,v)\big)$
is also a vertex of~$DG(t)$ and there is an edge in~$DG(t)$
from $\big(u,\tau+\delta(u,v)\big)$ to~$(v,\tau)$.
\item If~$(v,\tau)$ is a vertex of~$DG(t)$ and~$(u,v)$ is an edge in~$G$ such
that~$\tau+\delta(u,v)>t$, then $c_{u,v}$'s initial value~$x(u,v)$ is a vertex of~$DG(t)$ and there
is an edge in~$DG(t)$ from~$x(u,v)$ to~$(v,\tau)$.
\end{itemize}

Because all~$\delta(u,v)$ are strictly positive, the dependence graphs are finite and acyclic.
A vertex of~$DG(t)$ without incoming neighbors is a {\em leaf}, all others {\em intermediate vertices}.
A vertex of the form~$(i,\tau)$, with $i\in I$, is an {\em input leaf} and we call the time~$t-\tau$ 
the corresponding {\em measure point\/} for time~$t$.
If~$DG(t)=DG(\tilde{t})$, then the measure points for~$t$ are exactly the
measure points for~$\tilde{t}$ shifted by the difference $t-\tilde{t}$.
All leaves of~$DG(t)$ are either input leaves or elements of\/~$\Bool$
(initial values of channels).

\begin{figure}[hbt]
\centering
\begin{minipage}[b]{0.45\columnwidth}
\begin{tikzpicture}[circuit logic US, scale=1.4, large circuit symbols]
\matrix[column sep=15mm]
{
\node [or gate] (or) {OR}; &\\
};

\draw (or.input 1) -- ++(-0.7,0) node [pos=0.7,above] {$\delta=1$} node [pos=0.7,below] {$x=0$} node[left] (i) {$i$};
\draw (or.output) 
	-- ++(right:0.1)node[circle,inner sep=1pt,fill=black,draw] {}
	-- ++(down:0.6)
	-- ++(left:0.85) node[midway, above] {$\delta=2$} node[midway,below] {$x=0$}
	-- ++(up:0.4)
	|- (or.input 2)
	;
\draw (or.output)
	-- ++(0.7,0) node [pos=0.7,above] {$\delta=1$} node[pos=0.7,below] {$x=0$} 
	node[right] {$o$}
	;
\end{tikzpicture}
\caption{Example circuit}
\label{fig:Cex}
\end{minipage}%
%
\begin{minipage}[b]{0.5\columnwidth}
\begin{tikzpicture}[->,>=latex',shorten >=1pt,auto,node distance=0.5cm,semithick,scale=0.8]
  %
  %
  

  \path [draw] (0,0) circle (1mm) node [fill=white,xshift=-2mm,yshift=4mm] {$\li{o,0}$};
  \node (o0) at (0,0) {};

  \path [draw] (-1,0) circle (1mm) node [fill=white,below=1.5mm,xshift=1mm] {$\li{\mathrm{OR},1}$};
  \node (u1) at (-1,0) {};


  \path [draw] (-2,0.5) circle (1mm) node [fill=white,above=1.5mm,xshift=4mm] {$\li{\mathrm{OR},3}$};
  \node (u1e) at (-2,0.5) {};
  
  \path [fill] (-2,-0.5) circle (1mm) node [fill=white,left=1.5mm] {$\li{i,2}$};
  \node (i2) at (-2,-0.5) {};


  \path [draw] (-3,1) circle (1mm) node [fill=white,above=1.5mm]
{$\li{\mathrm{OR},5}$};
  \node (u12e) at (-3,1) {};

  \path [fill] (-3,0) circle (1mm) node [fill=white,left=1.5mm] {$\li{i,4}$};
  \node (i2e) at (-3,0) {};


  \path [fill] (-4,1.2) circle (1mm) node [fill=white,left=1.5mm] {$0$};
  \node (zero) at (-4,1.2) {};

  \path [fill] (-4,0.5) circle (1mm) node [fill=white,left=1.5mm] {$\li{i,6}$};
  \node (i22e) at (-4,0.5) {};

  %
  %

  \path [draw,->] (u1)   edge             (o0);
  \path [draw,->] (u1e)  edge             (u1);
  \path [draw,->] (i2)   edge             (u1);
  \path [draw,->] (u12e) edge             (u1e);
  \path [draw,->] (i2e)  edge             (u1e);
  \path [draw,->] (zero) edge             (u12e);
  \path [draw,->] (i22e) edge             (u12e);
  
  

\end{tikzpicture}
\caption{Example dependence graph $DG(6)$}\label{fig:depgr}
\end{minipage}
\end{figure}

As an example, consider the circuit shown in \figref{fig:Cex}.
The dependence graph $DG(6)$ is shown in \figref{fig:depgr}.
Leaves are depicted as filled nodes, while intermediate nodes are
     empty.
From the construction of the graph, we immediately see that in each execution
     the output signal value $s_o(6)$ only depends on the (input)
     signal values $s_i(4)$, $s_i(2)$, and
     $s_i(0)$.
Thus, in particular, $s_o(6)$ is the same for both input signals depicted in
     \figref{fig:input}.

\def\highV{0.5}
\def\offV{-2.8}
\begin{figure}[htb]
\centering
\begin{tikzpicture}[scale=1,>=latex']
  \path [draw,->] (0,0)--(0,{\highV+0.05}) node (xaxis) [above] {$s_i(t)$};
  \path [draw,->] (0,0)--(6.5,0)   node (yaxis) [above] {$t$};

  \path [draw,thin,-] (0,0.1) -- (0,-0.1) node [below] {$0$};
  \path [draw,thin,-] (1,0.1) -- (1,-0.1) node [below] {$1$};
  \path [draw,thin,-] (2,0.1) -- (2,-0.1) node [below] {$2$};
  \path [draw,thin,-] (3,0.1) -- (3,-0.1) node [below] {$3$};
  \path [draw,thin,-] (4,0.1) -- (4,-0.1) node [below] {$4$};
  \path [draw,thin,-] (5,0.1) -- (5,-0.1) node [below] {$5$};
  \path [draw,thin,-] (6,0.1) -- (6,-0.1) node [below] {$6$};

  \neline{0}{1.5}{0}
  \feline{1.5}{3.5}{\highV}
  \fnline{3.5}{6.1}{0}

  \tcrossAR{(0,0)}{$(i,6)$}
  \tcross{(2,\highV)}{$(i,4)$}
  \tcross{(4,0)}{$(i,2)$}
\end{tikzpicture}
\begin{tikzpicture}[scale=1,>=latex']
  \path [draw,->] (0,0)--(0,{\highV+0.05}) node (xaxis) [above] {$s_i(t)$};
  \path [draw,->] (0,0)--(6.5,0)   node (yaxis) [above] {$t$};

  \path [draw,thin,-] (0,0.1) -- (0,-0.1) node [below] {$0$};
  \path [draw,thin,-] (1,0.1) -- (1,-0.1) node [below] {$1$};
  \path [draw,thin,-] (2,0.1) -- (2,-0.1) node [below] {$2$};
  \path [draw,thin,-] (3,0.1) -- (3,-0.1) node [below] {$3$};
  \path [draw,thin,-] (4,0.1) -- (4,-0.1) node [below] {$4$};
  \path [draw,thin,-] (5,0.1) -- (5,-0.1) node [below] {$5$};
  \path [draw,thin,-] (6,0.1) -- (6,-0.1) node [below] {$6$};

  \neline{0}{0.5}{0}
  \feline{0.5}{1.1}{\highV}
  \feline{1.1}{1.5}{0}
  \feline{1.5}{2.5}{\highV}
  \feline{2.5}{4.5}{0}
  \feline{4.5}{5.5}{\highV}
  \fnline{5.5}{6.0}{0}

  \tcrossAR{(0,0)}{$(i,6)$}
  \tcross{(2,\highV)}{$(i,4)$}
  \tcross{(4,0)}{$(i,2)$}
\end{tikzpicture}
\caption{Input pulses with measure points ($\times$), labeled with
  the corresponding input leaf names.}\label{fig:input}
\end{figure}
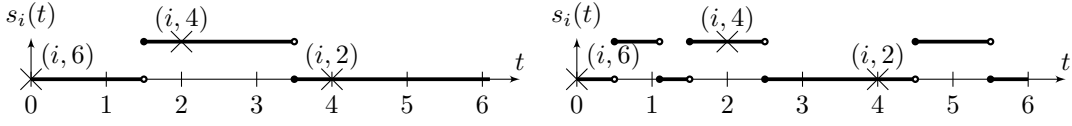

Generalizing the observations from the example, we thus observe:

\begin{lem}\label{lem:children}
The value of the output signal at time~$t$ only depends on the values of the
input signal at the measure points for time~$t$, according to $DG(t)$.

Furthermore, if~$DG(t)=DG(\tilde{t})$ and the values of input signals~$s_i$
and~$\tilde{s}_i$ coincide at the respective measure points for~$t$
and~$\tilde{t}$, then the respective output signals
fulfill~$s_o(t)=\tilde{s}_o(\tilde{t})$.
\end{lem}
\begin{proof}
For a path~$\pi$ in~$G$, denote by~$\delta(\pi)$ the sum of
delays~$\delta(u,v)$ over all edges~$(u,v)$ of~$\pi$.
For every vertex~$v$ of~$G$ and every time~$t\in\Time$, let~$\mathcal{P}(\to y, t)$ be the set of maximum length paths~$\pi$ ending in~$v$ such that~$\delta(\pi)\leq t$.

It is clear, by iterating Eq.~\eqref{eq:p:channel}, that the value of~$s_v(t)$ is uniquely
determined by the collection of values~$s_u\big(t - \delta(\pi) \big)$
where~$u$ is the start vertex of~$\pi\in\mathcal{P}(\to v, t)$.
Moreover, by maximality of~$\pi$, if~$u\neq i$, then $s_u\big(t - \delta(\pi) \big)$ only depends on the initial values of channels of incoming edges to~$u$.
Hence~$s_v(t)$ is uniquely determined by the collection of values~$s_i\big(t -
\delta(\pi) \big)$ where $\pi\in\mathcal{P}(\to y, t)$ starts at~$i$.
This holds in particular for~$v=o$.
\end{proof}

This lemma has as an immediate conseqeuence our remark at the end of
Section~\ref{subsec:exec}:

\begin{lem}\label{lem:executions}
If~$C$ is a circuit with only constant-delay channels, then for all
     assignments of input signals~$(s_i)_{i\in I}$ there exists a
     unique execution of~$C$ extending this assignment.
\end{lem}

Due to the fact that there are only finitely many measure points for a given
time~$t$, they are discrete and hence there is always a small margin until a
new measure point appears:

\begin{lem}\label{lem:shift}
For every~$t\in\Time$, there exists an~$\varepsilon>0$ such that
$DG(t)=DG(t+\varepsilon')$ for all $0\leq\varepsilon'\leq\varepsilon$.
\end{lem}
\begin{proof}
Let~$\varepsilon>0$ be strictly smaller than all positive values of the form $\delta(u,v)+\tau-t$
where~$(v,\tau)$ is an intermediate vertex of~$DG(t)$ and~$(u,v)$ is an edge in~$G$.
If no such intermediate vertex or edge exists, choose~$\varepsilon>0$ arbitrarily.

Let~$(v,\tau)$ be an intermediate vertex of~$DG(t)$ and $(u,v)$ be an edge in~$G$.
If $t + \varepsilon - \tau < \delta(u,v)$, then clearly $t - \tau < \delta(u,v)$, because~$\varepsilon>0$.
On the other hand, if~$t - \tau < \delta(u,v)$, then $\delta(u,v)+\tau-t$ is
positive and hence
$ \delta(u,v) > t + \varepsilon - \tau  $
by choice of~$\varepsilon$.
Thus, the conditions~$t - \tau < \delta(u,v)$
and~$t + \varepsilon - \tau < \delta(u,v)$ are equivalent.
This shows that the two dependence
graphs~$DG(t)$ and~$DG(t+\varepsilon)$ and hence all dependence graphs in between are equal.
\end{proof}

\subsection{Unsolvability Proof}

Assume by contradiction that~$C$ solves SPF.
By the nontriviality property (F4), there exists an input pulse such that the
corresponding output signal is non-zero, i.e., there exists an input pulse 
of some length and a time~$t$
such that the corresponding output signal's value at time~$t$ is~$1$.

By Lemma~\ref{lem:shift}, there exists an~$\varepsilon>0$ such that $DG(t) = DG(t+\varepsilon)$.
We may choose~$\varepsilon$ arbitrarily small, in particular
strictly smaller than all differences of distinct measure points for time~$t$.

Clearly, $DG(\tilde{t})=DG(t)$ for all times~$\tilde{t}$ between~$t$ and~$t+\varepsilon$, in particular, for $\tilde{t} = t + \varepsilon/2$.
Denote by~$\Delta$ the infimum of input pulse lengths (where all pulses 
start at the same time) such that the
corresponding output signal's value at time~$\tilde{t}$ is~$1$.
This infimum is finite by the choice of~$t$ and~$\tilde{t}$.
There hence exists an input pulse~$p$ with the above
property of length at most $\Delta+\varepsilon/4$.
We show that its corresponding output signal~$s_p$ contains a pulse of length
strictly less than~$\varepsilon$, in contradiction to the no short pulses property (F5).

Denote by~$S$ the time of~$p$'s upwards and by~$T$ the time of~$p$'s
downwards transition.
Now let~$p_+$ be the pulse whose upwards transition is at
time~$S$ and whose downwards transition is at time~$T-\varepsilon/2$.
If $S \geq T-\varepsilon/2$, then let~$p_+$ be the zero signal instead.
The length of~$p_+$ is either strictly less than~$\Delta$ or it is the zero
signal.
Hence, by the definition of the no-generation property~(F3),
its corresponding output signal's value at time~$\tilde{t}$ is~$0$.
This implies that there exists a measure point for time~$\tilde{t}$ 
within $[T-\varepsilon/2,T)$, because~$p$ and~$p_+$ coincide everywhere else
(see marked measure point on the right in \figref{fig:proof:p}).

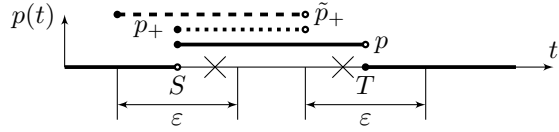
\begin{figure}
\centering
\begin{tikzpicture}[scale=1,>=latex']
  \path [draw,->] (0,0)--(0,{\highV+0.2}) node (xaxis) [left, yshift=-0.8] {$p(t)$};
  \path [draw,->] (0,0)--(6.5,0)   node (yaxis) [above] {$t$};

  \neline{0}{1.5}{0}
  \feline{1.5}{4.0}{\highV-0.2}
  \fnline{4.0}{6.0}{0}
  \node[below] at (1.5,0) {$S$};
  \node[below] at (4.0,0) {$T$};

  \path [draw,line width=0.5mm,-,dashed] (0.7,{\highV+0.2}) -- (3.2,{\highV+0.2}) node [midway,
above] {};%
  \fidot{(0.7,{\highV+0.2})}%
  \emdot{(3.2,{\highV+0.2})}%

  \path [draw,line width=0.5mm,-,dotted] (1.5,{\highV+0.0}) -- (3.2,{\highV+0.0}) node [midway,
above] {};%
  \fidot{(1.5,{\highV+0})}%
  \emdot{(3.2,{\highV+0})}%

  \node [right] at (4.0,{\highV-0.2}) {$p$};
  \node [right] at (0.8,{\highV+0}) {$p_+$};
  \node [right] at (3.2,{\highV+0.2}) {$\tilde{p}_+$};

  \tcross{(2.0,0)}{}
  \tcross{(3.7,0)}{}

  \path [draw,thin,-] (2.3,0) -- (2.3,-0.7) node [below] {};
  \path [draw,thin,-] (0.7,0) -- (0.7,-0.7) node [below] {};
  \path [draw,thin,-] (3.2,0) -- (3.2,-0.7) node [below] {};
  \path [draw,thin,-] (4.8,0) -- (4.8,-0.7) node [below] {};
  \path [draw,<->] (0.7,-0.5)-- node[below] {$\varepsilon$} (2.3,-0.5);
  \path [draw,<->] (3.2,-0.5)-- node[below] {$\varepsilon$} (4.8,-0.5);
\end{tikzpicture}
\caption{Input pulse~$p$, together with its derived pulses~$p_+$
and~$\tilde{p}_+$, and measure points for time~$\tilde{t}$.}
\label{fig:proof:p}
\end{figure}

Because we chose~$\varepsilon$ to be smaller than all differences of distinct
measure points for time~$t$ (and hence also for time~$\tilde{t}$), we see that there
is no measure point for~$\tilde{t}$ in the interval $[T,T+\varepsilon/2)$.

Likewise, by defining~$p_-$ as the pulse with upwards transition at
time~$S+\varepsilon/2$ and downwards transition at time~$T$, we infer that
there is one measure point for time~$\tilde{t}$ in the interval
$[S,S+\varepsilon/2)$ and there is no measure point for~$\tilde{t}$
in the interval $[S-\varepsilon/2,S)$ (see \figref{fig:proof:p}).

Now consider the pulse~$\tilde{p}_+$ generated by shifting pulse~$p$ into the
past by~$\varepsilon/2$, i.e., $\tilde{p}_+$'s upwards transition is at
time~$S-\varepsilon/2$ and its downwards transition is at~$T-\varepsilon/2$.
Because~$\tilde{p}_+$ coincides with~$p_+$ at all measure points
for~$\tilde{t}$, the output signal $s^{\tilde{p}_+}$
corresponding to~$\tilde{p}_+$ has value~$0$
at time~$\tilde{t}$.
Because $DG(\tilde{t}) = DG(\tilde{t}+\varepsilon/2)$, the second part of Lemma~\ref{lem:children} shows that
$s^{\tilde{p}_+}(\tilde{t}+\varepsilon/2)=0$.

Likewise, by considering~$p$ shifted into the future by~$\varepsilon/2$, we see
that also $s^{\tilde{p}_+}(\tilde{t}-\varepsilon/2)=0$.
But because $s^p(\tilde{t})=1$, this shows that the output signal~$s^p$
contains a pulse of length strictly less than~$\varepsilon$.
Since~$\varepsilon$ can be chosen arbitrarily small, this concludes the
proof.

\section{Bounded Single-History Channels}\label{sec:single}

This section formally introduces the notion of bounded single-history channels
in the
binary circuit model.
They are a generalization of constant-delay channels that cover all
existing channel models for binary circuit models we are aware of.

Intuitively, a bounded single-history channel propagates each event,
occurring at time~$t$, of the
input signal to an event at the output happening after some bounded
\emph{output-to-input}
delay $\delta(T)$, which depends on the \emph{input-to-previous-output}
delay $T=t-t'$. Note that $T$ is positive if the channel delay is short
compared to the input signal transition times, and negative otherwise.
\figref{fig:shc} illustrates this relation and the involved
delays. In case FIFO order would be invalidated, i.e., 
$t+\delta(T)\leq t'$, such that the next output event would not occur
after the previous one, both events annihilate.

There exist two variants of bounded single-history channels in the literature,
depending on whether the time of an annihilated event is remembered or not.
We dub these two variants {\em forgetful\/} and {\em non-forgetful\/}
bounded single-history channels, which we both formally define below.
At the end of this section, we give a list of channel models that are special
cases of our definition of bounded single-history channels.

\def\highV{0.5}
\def\offV{0.5}
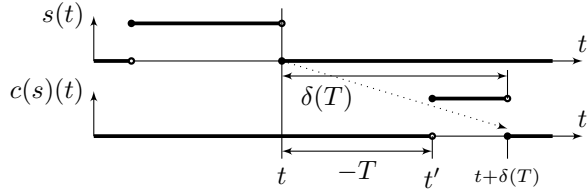
\begin{figure}[tbh]
\centering
\begin{tikzpicture}[scale=1,>=latex']
  \path [draw,->] (0,{\offV+\highV})--(0,{\offV+\highV+\highV+0.1}) node (xaxis) [left] {$s(t)$};
  \path [draw,->] (0,{\offV+\highV})--(6.5,{\offV+\highV})   node (yaxis) [above] {$t$};

  \neline{0}{0.5}{{\offV+\highV}}
  \feline{0.5}{2.5}{{\offV+\highV+\highV}}
  \fnline{2.5}{6.1}{{\offV+\highV}}
 
  \path [draw,thin,-] (2.5,1.7) -- (2.5,-0.3) node [below] {$t$};

  \path [draw,->] (0,0)--(0,{\highV+0.1}) node (xaxis) [left] {$c(s)(t)$};
  \path [draw,->] (0,0)--(6.5,0)   node (yaxis) [above] {$t$};

  \neline{0}{4.5}{0}
  \feline{4.5}{5.5}{\highV}
  \fnline{5.5}{6.1}{0}

  \path [draw,thin,-] (4.5,0) -- (4.5,-0.3) node [below] {$t'$};
  \path [draw,thin,-] (5.5,\highV) -- (5.5,{\highV+0.38});
  \path [draw,thin,-] (5.5,0) -- (5.5,-0.3) node [below] {$\scriptstyle t+\delta(T)$};
 
  \path [draw,<->] (2.5,{\highV+0.3})-- node[below,pos=0.2] {$\delta(T)$} (5.5,{\highV+0.3});
  \path [draw,<->] (2.5,-0.2)-- node[below] {$-T$} (4.5,-0.2);

  \path [draw,dotted,->] (2.5,{1.1-0.1}) -- (5.5,{0.0+0.1}) node {};

\end{tikzpicture}
\myvspace{-0.2cm}
\caption{Input/output signal of a bounded single-history channel, involving the input-to-previous-output delay $T$ and
the resulting output-to-input delay $\delta(T)$.}\label{fig:shc}\myvspace{-0.3cm}
\end{figure}

Formally, a {\em bounded single-history channel\/}~$c$ is characterized by an {\em initial
     value\/}~$x\in\Bool$,
     a nondecreasing {\em delay function\/}~$\delta:\IR\to\IR$ such
     that~$\delta(\infty)=\lim_{T\to\infty}\delta(T)$ is finite and positive,
and 
     the fact whether it is forgetful or not.
In the rest of the paper, we will drop the qualifier ``bounded'' when referring
to bounded single-history channels.
We detail the channel behavior in the next two subsections.

\subsection{Forgetful Single-History Channels}

This class of channels includes the classical inertial delay channels as used,
for example, in VHDL simulators~\cite{Ash08}.

Their behavior is defined by the following algorithm:
Let~$s$ be a signal. In case the channel's initial value~$x$ is equal
to the initial value of~$s$, or there is an event at time~$0$ in the
event list of~$s$, let the channel's {\em input list}
$\big((t_n,x_n)\big)_n$ be the event list of~$s$. Otherwise, let the
channel's input list be the event list of~$s$ with an additional event
at time~$0$ and value equal to the initial value of~$s$.
The algorithm iterates the input list and updates the {\em output list}, which
will define the channel's output signal~$c(s)$.

Initially, let $(-\infty,x)$ be the sole element of the output list.
In its $n$th iteration the algorithm considers input event~$(t_n,x_n)$ and
modifies the output list accordingly:
\begin{enumerate}
\item 
Denote by~$(t_n',x_n')$ the last event
in the output list.
If~$x_n = x_n'$, then input event $(t_n,x_n)$ has no effect:
Proceed to the $(n+1)$th iteration.

\item Otherwise, let
$ T_n =  t_n - t_n'$ be the difference of input and previous-output event times.\footnote{Note that $T_n=\infty$ is possible.
In this case $\delta(T_n) = \delta(\infty) =
\lim_{T\to\infty}\delta(T)$, which is finite by assumption.}

  If $t_n+\delta(T_n)>t_n'$, then add the event 
  $\big(t_n+\delta(T_n),x_n\big)$ to the output list.

  If $t_n+\delta(T_n)\leq t_n'$, then delete the event $(t_n',x_n')$ from the
output list.
\end{enumerate}

Note that the output sequence's first event is
     always~$(-\infty,x)$, all other events have positive times (since $\delta(\infty)>0$),
     its sequence of event times is strictly increasing, and its sequence of 
     values is alternating.

If the input list is finite, the algorithm halts.
If not, the output sequence nonetheless stabilizes in the sense that, for every
time~$t$, there exists some~$N$ such that all iterations with $n\geq
N$ make no changes to the output sequence at times $\leq t$.
The next lemma (Lemma~\ref{lem:def:chan:forget}) proves this property and makes the limit output list as $n$ tends
to infinity well-defined.
So, even if the input list is infinite, there exists a well-defined (infinite)
output list~$S$ that is the result of the described algorithm.
The channel's output signal~$c(s)$ is then defined by event
list~$S$:

\begin{definition}\label{def:shchannel:forget}
For input signal $s$, the output signal~$c(s)$ of the forgetful single-history
channel~$c$
is the signal whose event list is the list~$S$ as defined by the above
algorithm.
\end{definition}


\begin{lem}\label{lem:def:chan:forget}
Denote by~$S_n$ the output list after the $n$th iteration of the forgetful
channel algorithm, and by $S_n|t$ its
restriction to the
events at times at most~$t$.
For all~$t$ there exists an~$N$ such that~$S_n|t$ is constant 
for all $n\geq N$.
\end{lem}
\begin{proof}
The lemma is trivial if the input list is finite, so we assume it to be
infinite.

Because the sequence of input event times~$(t_n)$ tends to infinity,
there exists an~$N$ such that
\begin{equation}\label{eq:def:chan:n0}
 t_{N} \geq \max \big( t \,,\, t-\delta(-\delta(\infty)) \big) \enspace.
\end{equation}
We show by induction that $S_n|t = S_{N}|t$ for all $n\geq N$.
This is trivial for $n=N$, so let $n>N$.
Then $t_n > t_{N}$.

Let~$(t_n',x_n')$ be the last element in~$S_{n-1}$, and $T_n = t_n - t_n'$.
The case $x_n = x_n'$ is trivial, so let $x_n \neq x_n'$.
We distinguish two cases, depending on whether $\delta(T_n)>-T_n$ or not:

\smallskip

Case 1: $\delta(T_n)>-T_n$.
Because~$\delta$ is nondecreasing, $\delta(T_n)\leq \delta(\infty)$, and
hence $T_n> -\delta(\infty)$ and also $\delta(T_n) \geq
\delta(-\delta(\infty))$.
This implies $t_n + \delta(T_n) > t_{N} + \delta(-\delta(\infty))
\geq t$ by using~\eqref{eq:def:chan:n0}.
Hence $S_n | t = S_{n-1}|t = S_{N}|t$ by
the induction hypothesis.

\smallskip

Case 2: $\delta(T_n)\leq -T_n$.
We show that $t_n' > t$ by contradiction:
Let $t_n'\leq t$.
Then $T_n = t_n - t_n' > t_{N} - t \geq 0$, by
using~\eqref{eq:def:chan:n0}. From $\delta(\infty) > 0$, we
thus obtain $T_n > -\delta(\infty)$.
Hence $\delta(T_n) \geq \delta(-\delta(\infty))$ by monotonicity of~$\delta$.
By assumption, $\delta(-\delta(\infty))\leq\delta(T_n) \leq -T_n = t_n' - t_n$, which
implies $t_n \leq t_n' - \delta(-\delta(\infty))$, i.e., $t_{N} < t - \delta(-\delta(\infty))$.
This is a contradiction to~\eqref{eq:def:chan:n0}, which shows that $t_n'>t$.
Hence
$ S_n | t = S_{n-1} | t = S_{N} | t$ 
by the induction hypothesis.
\end{proof}

\subsection{Non-Forgetful Single-History Channels}

The PID channel introduced by Bellido-D\'{i}az et al.~\cite{BDJCAVH00}
is not covered by the above forgetful single-history channels, since
it has been designed to reasonably match analog RC waveforms: 
Analog signals like exponential functions do not ``forget'' 
sub-threshold pulses. Hence, they cannot be modeled via delay functions $\delta(T)$ 
that depend on the input-to-previous output delay~$T$. To also
cover the PID model, we hence introduce non-forgetful 
single-history channels, the delay function of which may also
depend on the last annihilated event.

The output-eventlist generation algorithm for non-forgetful channels thus maintains an additional
variable~$r$, which, in each iteration, contains the time of the \emph{potential
output event} considered in the last iteration. Note that this approach 
was already
used in the PID-channel-model by Bellido-D\'\i az et al.~\cite[Fig.~13]{BDJCAVH00}. Similar to
the forgetful case, it determines the output 
signal~$c(s)$ of a non-forgetful
single-history channel~$c$, given input signal~$s$ with input
event list~$\big( (t_n,x_n) \big)_n$
as follows:

Initially, the output list contains the sole element~$(-\infty,x)$
and~$r=r_{-1}=-\infty$.
In its $n$th iteration, the algorithm considers input event~$(t_n,x_n)$ and
modifies the output list accordingly:
\begin{enumerate}
\item 
Denote by~$(t_n',x_n')$ the last event
in the output list.
If~$x_n = x_n'$, then input event $(t_n,x_n)$ has no effect:
Proceed to the $(n+1)$th iteration.
\item 
Otherwise, let
$ T_n =  t_n - r_{n-1}$ be the difference of input and most recent potential output event times
and set $r_n = t_n + \delta(T_n)$.

  If $t_n+\delta(T_n)>r_{n-1}$, then add the event 
  $\big(t_n+\delta(T_n),x_n\big)$ to the output list.

  If $t_n+\delta(T_n)\leq r_{n-1}$, then delete the event $(t_n',x_n')$ from the
output list.
\end{enumerate}

We first show that if event $(t_n',x_n')$ is deleted in the~$n$th
     iteration, then $r_{n-1} = t_n'$.
\begin{proof}
Assume by contradiction that this is not the case, and let~$n$ be the
     first iteration where the statement is violated.
Then it must hold that $n\ge 2$, as in iteration~$n-2$ some event
     $(\tau,x_{n-2})$ must have been added to the output list that was deleted in
     iteration~$n-1$, due to $\tau' = t_{n-1}+\delta(T_{n-1})\leq
     r_{n-2} = \tau$.
Furthermore, in iteration~$n$, our assumption of deleting some event with a time
different from~$r_{n-1}=\tau'$ implies $\tau'' = t_{n}+\delta(T_{n})\leq \tau'$.
However, from $t_{n-1} < t_n$, $\tau \ge \tau'$ and monotonicity
     of~$\delta$, $t_{n-1}+\delta(t_{n-1}-\tau) < t+\delta(t-\tau')$,
     i.e., $\tau' < \tau''$, which provides the required contradiction.
\end{proof}

Thus, an event is either deleted in the next iteration, or never
deleted. The output sequence's first event~$(-\infty,x)$ is obviously
never deleted.

By analogous arguments, one can show that the sequence of event times
is strictly increasing, with an alternating sequence of values.
Unlike in the case of forgetful channels, however, the eventlist generation
algorithm may produce events with \emph{finite} negative times that
will be removed from the final output.
In case the input list is finite, the algorithm clearly halts.
If not, we again have the same stabilization property as for forgetful
single-history channels, which we will provide in
Lemma~\ref{lem:def:chan:nonforget} below.
Thus the algorithm's final output list~$S$ is again well-defined and we can define:

\begin{definition}\label{def:shchannel:nonforget}
For input signal $s$, the output signal~$c(s)$ of the forgetful single-history
channel~$c$
is the signal whose event list is the list~$S$ as defined by the above
algorithm.
after deleting all events with finite negative times and 
the first non-negative
time event
if its value is equal to the channel's initial value~$x$.
\end{definition}

\begin{lem}\label{lem:def:chan:nonforget}
Denote by~$S_n$ the output list after the $n$-th iteration of the forgetful
channel algorithm, and by $S_n|t$ its
restriction to the
events at times at most~$t$.
For all~$t$, there exists an~$N$ such that~$S_n|t$ is constant 
for all $n\geq N$.
\end{lem}
\begin{proof}
The lemma follows from the fact that an event can only be deleted one
     iteration after it was added to the output list, and the fact
     that in each iteration~$n$, $T_n > -\delta(\infty)$ and thus
     $t_n+\delta(T_n)$ is lower bounded by
     $t_n+\lim_{t\to 0^+}\delta(-\delta(\infty)+t)$.
\end{proof}

\subsection{Examples of Single-History Channels}

Below, we summarize how the existing binary-value models
are mapped to our single-history channels:

\begin{enumerate}[1)]
\item A classic {\em pure-delay channel\/} is a single-history channel whose
     delay function~$\delta$ is constant and positive.
The behavior of a pure-delay channel does not depend on
     whether it is forgetful or not.

\item An inertial channel is a forgetful single-history channel whose delay function~$\delta$ is of the form 
\[ \delta(T) = \begin{cases} \delta_0 & \text{ if } T>T_0 \\ -T_0 & \text{ if } T\leq T_0 \end{cases} \]
for parameters~$\delta_0>0$ and~$T_0>-\delta_0$.
An inertial channel filters an incoming pulse if and only if its pulse length is less or equal to~$T_0+\delta_0$; otherwise, it is forwarded with delay~$\delta_0$.

\item The PID-channels of Bellido-D\'\i az et al.~\cite{BDJCAVH00} 
are non-forgetful with
delay function
\begin{equation}
\delta(T) = t_{p0} \cdot \left( 1 - e^{-(T - T_0)/\tau} \right)
\end{equation}
for certain (measured) positive parameters~$t_{p0}$,~$\tau$, and~$T_0$.
Note that $\delta(T_0)=0$, $\lim_{t\to\infty}\delta(T)=t_{p0}$, and
$\frac{d\delta(T)}{dT}|_{T=0}=t_{p0}/\tau$ here.
\end{enumerate}

\section{Bounded Short-Pulse Filtration with One Non-Constant Delay Channel}\label{sec:PD}

In this section we prove that bounded SPF is solvable as soon as there is a
single non-constant-delay single-history channel available.
More specifically, we show that, given a single-history channel with
non-constant delay, there exists a circuit that uses only constant-delay
channels apart from the given non-constant channel that solves bounded SPF.
Different circuits, and hence proofs, are used in for different types of 
channels.

For a single-history channel with delay function~$\delta$, let
     $\delta_\infty = \delta(\infty) = \lim_{t\to\infty}\delta(t)$ 
with $0<\delta_{\infty}<\infty$. The right limit of~$\delta$ 
at~$-\delta_{\infty}$ is denoted by
$\delta_{\inf} = \lim_{t\to 0^+}\delta(-\delta_\infty+t)$; note
that $\delta_{\inf}=-\infty$ is allowed here.

In the rest of this section, let~$c^*$ be a single-history channel
     that is not a constant-delay channel as defined in Section~\ref{sec:channels}.
This is equivalent to saying that its delay function~$\delta$ is
     non-constant for $T>-\delta_\infty$, because $T_n > -
     \delta_\infty$ in every step of the channel algorithm:   

\begin{lem}\label{lem:char:const}
A single-history channel with delay function~$\delta$
is a constant-delay channel if and only if $\delta$ is constant in the
open interval $( -\delta_\infty , \infty)$.
\end{lem}
%
%

Note that $\delta_{\inf}<\delta_{\infty}$ in case of a non-constant
delay channel.
From the fact that $-\delta_\infty < T_n \leq \infty$ in every step of the
     channel algorithm, we also obtain:

\begin{lem}\label{lem:char:delaybounds}
All events 
in the event list of a single-history
     channel's input signal are delayed by times within
     $[\delta_{\inf},\delta_\infty]$.
\end{lem}

\subsection{Forgetful Channels}

In this subsection, assume that~$c^*$ is forgetful.
Consider circuit~$C_\mathrm{ff}$ depicted in \figref{fig:Cinv},
     which contains channel~$c^*$ as well as two constant-delay channels.
For the moment assume that the initial value of~$c^*$ is~$0$.
We will show at the end of this subsection that bounded SPF is also
solvable with~$c^*$ if its initial value is~$1$.

\begin{figure}[bt]
\centering
\begin{tikzpicture}[circuit logic US, scale=1.1, large circuit symbols]
\matrix[column sep=15mm]
{
& \node [or gate] (or) {OR}; &  \\
};

\draw (or.input 1) -- ++(-2.5,0) node [pos=0.7,above] {$\delta=1$} node [pos=0.7,below] {$x=0$} node[left] (i) {$i$};
\draw (or.output) 
	-- ++(right:0.1)node[circle,inner sep=1pt,fill=black,draw] {}
	-- ++(down:0.6)
	-- ++(left:2.0) node[pos=0.7, above] {$\delta=\varepsilon$} node[pos=0.7,below] {$x=0$}
	-- ++(up:0.4)
	|- (or.input 2)
	;
\draw (or.output)
	-- ++(1.3,0) node [midway,above] {$c^*$} node[midway,below] {} 
	node[right] {$o$}
	;

\end{tikzpicture}
\myvspace{-0.2cm}
\caption{Circuit $C_\mathrm{ff}$.}\label{fig:Cinv}\myvspace{-0.2cm}
\end{figure}

It remains to describe how to choose delay parameter~$\varepsilon>0$.
We will show in the following that for each non-constant-delay forgetful single-history
channel~$c$ there exists a~$\gamma(c)>0$ such that~$c(s)$ is the zero signal
whenever~$s$ is a pulse of length less than~$\gamma(c)$.
More generally we will show that, if signal~$s$ does not contain pulses of
length greater or equal to~$\gamma(c)$, then~$c(s)$ is the zero signal.
We then choose 
$0 < \varepsilon < \gamma(c^*)$
for the delay parameter~$\varepsilon$ in circuit~$C_\mathrm{ff}$.

If the input signal of circuit~$C_\mathrm{ff}$ is a pulse of length at
least~$\varepsilon$, then the signal~$s_{OR}$ at the OR gate is eventually
stable~$1$ because of the $\varepsilon$-delay feedback loop, and hence the
circuit's output signal is eventually stable~$1$.
If the circuit's input signal is a pulse of length~$\Delta<\varepsilon$,
then~$s_{OR}$ only contains pulses of length~$\Delta<\gamma(c^*)$, from which
it follows that the circuit's output signal is zero.

Let~$\delta$ be the delay function of a single-history channel~$c$.
We define:
\begin{equation}
\gamma(c) = \inf \left\{ \Delta>0 \mid \Delta - \delta_\infty + \delta\big(
\Delta - \delta_\infty \big) > 0 \right\}
\end{equation}
We will prove~$\gamma(c^*)>0$ in Lemma~\ref{lem:constant:equiv}.
Before characterizing the non-constant-delay channels as those~$c$ with~$\gamma(c)>0$, we need a preliminary lemma on pulse-filtration
properties of non-constant-delay channels.

\begin{lem}\label{lem:gamma:immediate}
Let~$c$ be a non-constant-delay single-history channel with initial value~$0$.
If~$s$ is a pulse of length less than~$\gamma(c)$, then~$c(s)$ is zero.
\end{lem}
\begin{proof}
The event list of~$s$ consists of two events~$(S,1)$ and~$(T,0)$, possibly preceded by an additional event~$(0,0)$, depending on whether~$S=0$ or~$S>0$.
Because the initial value of~$c$ is~$0$, we may assume without loss of
generality that the sequence consists of only these two events.

After iteration~$n=0$ of the channel-defining algorithm, the output list is
equal to 
$\big( (-\infty, 0),  (S + \delta_\infty, 1)  \big)$.
Hence, in iteration~$n=1$,
\[ T_1 = T - S - \delta_\infty < \gamma(c) - \delta_\infty \enspace, \]
i.e., $T_1+\delta_\infty < \gamma(c)$.
By definition of $\gamma(c)$, this implies
\[ (T_1+\delta_\infty)-\delta_\infty
+\delta((T_1+\delta_\infty)-\delta_\infty) \leq 0 \enspace, \]
and thus
$T_1 + \delta(T_1) \leq 0$. Thus, the event $(S+\delta_\infty,1)$ gets
removed from the output list and the output signal is the constant-zero signal.
\end{proof}

\begin{lem}\label{lem:constant:equiv}
Let~$c$ be a single-history channel with initial value~$0$.
The following statements are equivalent:
\begin{enumerate}
\item $c$ is not a constant-delay channel.
\item There exist a pulse~$s$ such that $c(s)$ is the zero
signal.
\item $\gamma(c) > 0$
\end{enumerate}
\end{lem}
\begin{proof}
Let~$\delta$ be the delay function of~$c$.
If~$s$ is a pulse of length~$\Delta$, then~$c(s)$ is zero if and only if 
\begin{equation*}
\Delta - \delta_\infty + \delta\big( \Delta - \delta_\infty \big) \leq 0\enspace.
\end{equation*}
This implies $\gamma(c)\geq\Delta$ and hence
establishes the equivalence of (2) and (3).
If we can show that~$c$ is not a constant-delay channel if and only if
\begin{equation}
\exists\varepsilon>0\,:\ \delta( -\delta_\infty + \varepsilon ) \leq \delta_\infty - \varepsilon\enspace,\label{eq:equiv1}
\end{equation}
then we can choose~$\Delta= \varepsilon$, concluding the proof.

The sufficiency of Eq.~\eqref{eq:equiv1} for $c$ not being a constant-delay
channel is immediate.
To prove the necessity of Eq.~\eqref{eq:equiv1}, assume that~$c$ is not a
constant-delay channel.
Then there exist $\beta,\beta'>0$ such that    
$  \delta(\beta-\delta_\infty) < \delta(\beta'-\delta_\infty)$
and since $\delta$ is nondecreasing,
$ \delta(\beta-\delta_\infty) < \delta_\infty$.
Thus, there exists a $z > 0$, such that,
\begin{align}
  \delta(\beta-\delta_\infty) \leq \delta_\infty-z\enspace.\label{eq:z}
\end{align}

There are two cases for $z$:
If $\beta \leq z$, we obtain from Eq.~\eqref{eq:z} that
$\delta(\beta-\delta_\infty) \leq \delta_\infty-\beta$.
Choosing $\varepsilon=\beta$ shows that Eq.~\eqref{eq:equiv1} holds.
Otherwise, i.e., if $\beta > z$, we obtain from Eq.~\eqref{eq:z} and the
     fact that $\delta$ is nondecreasing
\begin{align}
  \delta(z-\delta_\infty) \leq \delta(\beta-\delta_\infty) \leq \delta_\infty-z\enspace.\notag
\end{align}
Choosing $\varepsilon=z$ shows that Eq.~\eqref{eq:equiv1} holds.
\end{proof}

Note that, while Lemmas~\ref{lem:gamma:immediate} and~\ref{lem:constant:equiv} hold for both forgetful and
non-forgetful single-history channels, the following lemma does fundamentally
not hold for arbitrary non-forgetful channels.

\begin{lem}\label{lem:train}
Let~$c$ be a non-constant-delay forgetful single-history channel with initial value~$0$.
Let~$s$ be a signal that does not contain pulses of length greater or equal
to~$\gamma(c)$ and that is not eventually equal to~$1$.
Then~$c(s)$ is the zero signal.
\end{lem}
\begin{proof}
The lemma is proved by inductively repeating the proof of
Lemma~\ref{lem:gamma:immediate} for all pulses contained in $s$.
\end{proof}

\begin{lem}\label{lem:c0}
Circuit~$C_\mathrm{ff}$ solves bounded SPF.
\end{lem}
\begin{proof}
We first note that, given an input signal, there is a unique execution for
circuit~$C_\mathrm{ff}$ according to Lemma~\ref{lem:executions}, 
because the sole non-constant channel~$c^*$ is not
part of a feedback loop.

The well-formedness properties (F1) and (F2) of SPF are hence fulfilled.
The non-generation property (F3) is also obvious.

If the input signal is a pulse of length at least~$\varepsilon$,
then~$s_\mathrm{OR}(t) = 1$ for all $t\geq S+1$, and hence~$s_o(t) = 1$
for all~$t \geq S+1+\delta^*(\infty)$.
In particular, this shows the nontriviality property (F4).

If the input signal is a pulse of length less than $\varepsilon$,
then~$s_\mathrm{OR}(t)$ 
only contains pulses of lengths less than~$\varepsilon$, hence less
than~$\gamma(c^*)$ by the choice of $\varepsilon$. 
By Lemma~\ref{lem:train}, the output signal is zero in this case.
This, together with the above, shows (F5) and (F6).
\end{proof}

It remains to show that assuming $c^*$ to have initial value~$0$ is
     is not restricting: If its initial value is~$1$ we
     modify circuit~$C_\mathrm{ff}$ by adding an inverter before and
     after channel~$c^*$.
A proof analogous to Lemma~\ref{lem:c0}'s yields:

\begin{thm}\label{thm:cpd2}
Let~$c^*$ be a non-constant-delay forgetful single-history channel.
Then there exists a circuit solving bounded SPF whose channels are either
constant-delay channels or~$c^*$.
\end{thm}

\subsection{Non-Forgetful Channels}

Theorem~\ref{thm:nonforgetful} reveals 
that a single non-constant-delay {\em non-forgetful\/} 
single-history channel~$c^*$ (with initial value $0$) also 
allows to solve bounded SPF:

\begin{thm}\label{thm:nonforgetful}
Let~$c^*$ be a non-constant-delay non-forgetful single history channel with
     initial value~$0$.
Then there exists a circuit solving SPF whose channels are all either
     constant-delay channels or~$c^*$.
\end{thm}     

Let $\delta$ be the delay function of~$c^*$. Recall from Lemma~\ref{lem:char:const}
that $\delta_{\inf}<\delta_{\infty}$, since $\delta$ is
     non-decreasing and not constant.
We distinguish three cases for function~$\delta$ with respect to its
     behavior at~$-\delta_{\inf}$.

\begin{itemize}
\item[1.] There exists a $t > -\delta_{\inf}$ such
     that $\delta(t) < \delta_\infty$.

\item[2.] $\delta(t) = \delta_\infty$ for all $t
     > -\delta_{\inf}$, and
\begin{itemize}
\item[2.1] $\delta$ is continuous at~$-\delta_{\inf}$, i.e., at~$-\delta_{\inf}$
    its left limit $\lim_{t\to
     0^-}\delta(-\delta_{\inf}+t)$
     equals its right limit~$\delta_{\infty}$.

\item[2.2] $\delta$ is non-continuous at~$-\delta_{\inf}$, i.e.,
     $\delta^- = \lim_{t\to 0^-}\delta(-\delta_{\inf}+t) <
     \delta_\infty$.
\end{itemize}
\end{itemize}

For Cases~1 and~2.1, we show that circuit~$C_\mathrm{NF}$ depicted in
     \figref{fig:nf} solves bounded SPF.
All its clocks~$CLK_{A/C/F}$ produce a signal with period $A+B+C+D$,
     where parameters~$A$ to~$D$ are chosen later on in accordance
     with~$\delta$.
Let $\tau_k = k(A+B+C+D)$ denote the beginning of the $k$-th \emph{round}, for $k\ge 0$.
Clock~$CLK_C$ is designed such that its output signal is~$0$ during
     $[\tau_k,\tau_k+A+B) \cup [\tau_k+A+B+C,\tau_{k+1})$ and~$1$
     during $[\tau_k+A+B,\tau_k+A+B+C)$.
Such a clock can easily be built from constant-delay channels and inverters
     only.
Clock~$CLK_A$'s output signal is~$1$ during $[\tau_k,\tau_k+A)$
     and~$0$ during $[\tau_k+A,\tau_{k+1})$.
The output signal of~$CLK_F$ is~$0$ during $[\tau_k,\tau_k+E)\cup
     [\tau_k+E+F,\tau_{k+1})$ and~$1$ during $[\tau_k+E,\tau_k+E+F)$.
Again, $E$ and~$F$ are chosen later on in accordance with~$\delta$.

Abbreviating $t_k = \tau_k+2$, we 
observe that circuit~$C_\mathrm{NF}$ generates a signal~$s_{OR}$ at
     the input of channel~$c^*$, which is the OR of two subsignals
that consist of four {\em phases} within
     time~$[t_k,t_{k+1})$, $k\ge 0$ (i.e., per round):
     Phase~$A$ (of round~$k$) denotes the interval of
     times~$[t_k,t_k+A)$, phase~$B$ the interval~$[t_k+A,t_k+A+B)$,
     phase~$C$ the interval~$[t_k+A+B,t_k+A+B+C)$ and phase~$D$ the
     interval~$[t_k+A+B+C,t_k+A+B+C+D)$.
The value of~$s_{OR}$ is~$1$ during phase~$A$, and~$0$ during
     phases~$B$ and~$D$.
During phase~$C$ it is either~$0$ or contains a pulse, depending on
     signal~$i$.
Analogously, we define output phase~$F$ (of round~$k$) as the interval
     of times~$[t_{k}+E,t_{k}+E+F)$. Note that phase~E and~F of round
$k$ follow phase~D of round~$k$, and overlap with phase~A of round $k+1$.

Informally, for Cases~1 and~2.1, circuit~$C_\mathrm{NF}$ solves bounded SPF
according to the following reasoning: Properties~(F1) and~(F2) trivially
     hold for circuit~$C_\mathrm{NF}$.
Clearly, if the circuit's input signal is~$0$, then the channel's
     input signal~$s_{OR}$ is~$0$ during phase~$C$ of all rounds~$k\ge
     0$.
Subsequently, we will prove that if this is the case, then the
     channel's output signal~$c^*(s_{OR})$ during phase~$F$ is~$0$ for
     all rounds~$k\ge 0$.
Since phase~$F$ is the only phase where $o$ could possibly produce a
     non-$0$ output due to the AND gate, both~(F3) and (F5) follow.
Property~(F4) is implied by the fact that there exists an input
     signal~$i$ such that~$s_{OR}$ contains a pulse during phase~$C$
     of some round~$k\ge 0$.
We will prove below that if this is the case, then the  channel's
     output signal is~$1$ during phase~$F$ of round~$k+1$.
Essentially, this follows from a
reduced delay of the rising transition at the end of phase~D, caused by
not forgetting the (cancelled) pulse in phase~C.
From this and the fact that all delays are bounded, (F6)
     follows.

\medskip

\noindent {\bf Case~1.} In this case, we choose
\begin{enumerate}
\item[(i)] $C > 0$, $D>0$ and $0 < \Delta < \delta_\infty$ such that
     $\delta(C+D-\delta_{\inf}) \leq \delta_\infty - \Delta$.
Such values for $C$, $D$ and $\Delta$ exist, because of the assumption
     of Case~1.

\item[(ii)] $\varepsilon > 0$, $\varepsilon' > 0$ and $C > 0$ small enough such
     that $\delta_\infty - \varepsilon' \ge \delta_{\inf} +
     \varepsilon + C$ and $\varepsilon' < \Delta/4$.

\item[(iii)] $C > 0$ and $\varepsilon' > 0$ small enough such that
     $\delta(C+\varepsilon'-\delta_\infty) \leq \delta_{\inf} +
     \varepsilon$.

\item[(iv)] $A=B > \max(\varepsilon',\Delta,\delta_\infty-\delta_{\inf})$
     and large enough such that
     $\delta(A-\delta_\infty) \ge \delta_\infty - \varepsilon'$.

\item[(v)] $E = \delta_\infty-\Delta$ and $F = \Delta/2$.

\end{enumerate}
It is easy to check that Assumptions (i)--(v) are compatible with
each other.

Figures~\ref{fig:non-const:0} and~\ref{fig:non-const:1} depict
     signal~$s_{OR}$ in absence and presence of a pulse.
We will first show that the channel's output signal~$c^*(s_{OR})$
     has value~$0$ during output phase~$F$ of round~$0$.

\begin{proof}
The signal is depicted in \figref{fig:non-const:0}:
Signal~$s_{OR}$'s transition to value~$1$ at time~$t_0$ is delayed by
     $c^*$ by~$\delta_0 = \delta_\infty > 0$.
Its next transition back to value~$0$ at time~$t_0+A$ is delayed by,
     say, $\delta_1$.
Because of Lemma~\ref{lem:char:delaybounds}, $\delta_1 \ge
     \delta_{\inf}$.
From this and Assumption~(iv) on~$A$,  
\begin{align*}
A + \delta_1 > (\delta_\infty-\delta_{\inf}) +
\delta_{\inf} = \delta_0\enspace.
\end{align*}
It follows that output $c^*(s_{OR})$'s transition to~$0$ does not
     cancel~$c^*(s_{OR})$'s transition to~$1$ from before.
All of $s_{OR}$'s following transitions occur at times at least
     $t_0+A+B$, and by~(iv), at times greater than
     $t_0+\delta_\infty-\delta_{\inf}$.
Since all these transitions are delayed by at least
     $\delta_{\inf}$ time, 
none of them can cancel $c^*(s_{OR})$'s transition to~$1$ at
     time~$t_0+\delta_\infty$ either.
Since channel~$c^*$ has initial value~$0$, it follows that its output has
     value~$0$ during $[0,t_0+\delta_\infty)$.
Since
\begin{align*}
t_0+\delta_\infty > t_0+\delta_\infty-\Delta/2 = t_0 +E+F\enspace,
\end{align*}
the channel's output indeed has value~$0$ during output
     phase~$F$ of round~$0$.
\end{proof}

We next show, for $k\ge 0$, that if signal~$s_{OR}$ does
     not contain a pulse within phase~$C$ of round~$k$,
     signal~$c^*(s_{OR})$ has value~$0$ during output phase~$F$ of
     round~$k+1$.

%
%

\def\highV{0.5}
\def\offV{-1.8}
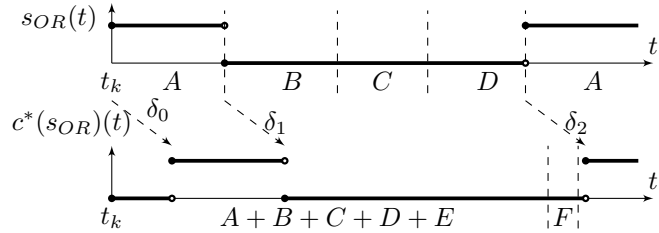
\begin{figure}[ht]
\centering
\begin{tikzpicture}[scale=1,>=latex']
  \path [draw,->] (0,0)--(0,{\highV+0.3}) node (xaxis) [above,xshift=-7mm,yshift=-5mm] {$s_{OR}(t)$};
  \path [draw,->] (0,0)--(7.2,0)   node (yaxis) [above] {$t$};

  \feline{0}{1.5}{\highV}
  \feline{1.5}{5.5}{0}
  \fnline{5.5}{7}{\highV}

  \draw (0,0) node[below] {$t_k$};

  \draw[dashed] (1.5,-0.4) -- ++(0,{\highV+0.6});
  \node[below] at (0.8,0) {$A$};

  \draw[dashed] (3.0,-0.4) -- ++(0,{\highV+0.6});
  \node[below] at (2.4,0) {$B$};

  \draw[dashed] (4.2,-0.4) -- ++(0,{\highV+0.6});
  \node[below] at (3.6,0) {$C$};

  \draw[dashed] (5.5,-0.4) -- ++(0,{\highV+0.6});
  \node[below] at (5,0) {$D$};

  \node[below] at (6.4,0) {$A$};

  \path [draw,->] (0,{\offV})--(0,{\offV+\highV+0.2}) node (xaxis) [above,xshift=-15pt] {$c^*(s_{OR})(t)$};
  \path [draw,->] (0,{\offV})--(7.2,{\offV}) node (yaxis) [above] {$t$};

  \feline{0}{0.8}{\offV}
  \feline{0.8}{2.3}{{\offV+\highV}} 
  \feline{2.3}{6.3}{\offV} 
  \fnline{6.3}{7.0}{{\offV+\highV}}

  \draw (0,\offV) node[below] {$t_k$};

  \draw[dashed] (5.8,{\offV-0.4}) -- ++(0,{\highV+0.6});
  \node[below] at (3.0,\offV) {$A+B+C+D+E$};

  \draw[dashed] (6.2,{\offV-0.4}) -- ++(0,{\highV+0.6});
  \node[below] at (6.0,\offV) {$F$};

  \draw[dashed,->] (0,-0.5) -- (0.8,-1.1)
    node[pos=0.2,right,xshift=5pt] {$\delta_0$};
  \draw[dashed,->] (1.5,-0.5) -- (2.3,-1.1)
    node[midway,right] {$\delta_1$};
  \draw[dashed,->] (5.5,-0.5) -- (6.3,-1.1)
    node[midway,right] {$\delta_2$};

\end{tikzpicture}
\caption{Case 1: Input and Output of channel~$c^*$ in circuit~$C_\mathrm{NF}$
     if phase~$C$ does not contain a pulse.}\label{fig:non-const:0} 
\end{figure}

\begin{proof}
Assume the input signal~$s_{OR}$ of channel~$c^*$ does not contain a
     pulse within phase~$C$ of round~$k$.
The signal is depicted in \figref{fig:non-const:0}.

Signal~$s_{OR}$'s transition to value~$1$ at time~$t_k$ is delayed by $c^*$ 
     by~$\delta_0 \leq \delta_\infty$.

There is no transition of~$s_{OR}$ before $s_{OR}$'s transition back
     to value~$0$ at time~$t_k+A$.
Let~$\delta_1$ be its delay.
Because of~(iv), and~$\delta$ being non-decreasing, 
$A + \delta_1 > (\delta_\infty-\delta_{\inf}) +
\delta_{\inf}$.
Thus, and because transitions are delayed by at least
     $\delta_{\inf}$, none of the transitions from
     time~$t_k+A$ on may cancel~$c^*(s_{OR})$'s transition to~$1$ at
     time~$t_k+\delta_0$.

The transition of~$s_{OR}$ to value~$1$ at time~$t_{k+1} =
     t_k+A+B+C+D$ is delayed by $\delta_2$, where  
\begin{align}
\delta_2 = \delta(B+C+D-\delta_1) \ge \delta(B-\delta_{\infty}) \ge \delta_\infty-\varepsilon'\enspace,\label{eq:delta_eps_pr}
\end{align}
because of Assumption~(iv). Together with~(ii) this yields
\begin{align}
\delta_2 > \delta_\infty - \Delta/4\enspace.\label{eq:delta2bound}
\end{align}
It will thus not occur at output~$c^*(s_{OR})$ before
     time~$t_{k+1}+\delta_\infty - \Delta/4$, and thus, by~(v), not before the end of output phase~$F$ of
     round~$k+1$ at time~$t_{k+1}+\delta_\infty - \Delta/2$.

Furthermore, from \eqref{eq:delta_eps_pr} and~(iv),
\begin{align*}
B+C+D + \delta_2 > \delta_\infty \ge \delta_1\enspace,
\end{align*}
because~(iv) in particular implies $B > \varepsilon'$.
It follows that output~$c^*(s_{OR})$'s transition to~$1$ does not
     cancel~$c^*(s_{OR})$'s transition to~$0$ at
     time~$t_k+A+\delta_1$.
All $s_{OR}$'s subsequent transitions occur at earliest at time
     $t_{k+1} + A > t_{k+1} + \delta_{\infty} - \delta_{\inf}$, by~(iv) and the fact that
     they are delayed by at least $\delta_{\inf}$, hence cannot
     cancel $c^*(s_{OR})$'s transition to~$1$ at time
     $t_{k+1}+\delta_2$.
Thus, $c^*(s_{OR})$ has value~$0$ during $[t_k+A+\delta_1,
     t_{k+1}+\delta_2)$. Together with (\ref{eq:delta2bound}), 
this implies that $c^*(s_{OR})$'s
     value is~$0$ during phase~$F$ of round~$k+1$.
\end{proof}



We now show, for $k\ge 0$, that if signal~$s_{OR}$ contains a pulse
     within phase~$C$ of round~$k$, signal~$c^*(s_{OR})$ has value~$1$
     during output phase~$F$ of round~$k+1$.

\begin{proof}
Assume the input signal~$s_{OR}$ of channel~$c^*$ contains a pulse within phase~$C$ of
     round~$k$.
The signal is depicted in \figref{fig:non-const:1}.

Signal~$s_{OR}$'s transition to value~$1$ at time~$t_k$ is delayed
     by~$\delta_0 \leq \delta_\infty$.
By the same arguments as in the proof before, it is not canceled by any
     following transition.

Signal~$s_{OR}$'s transition to~$0$ at time~$t_k+A$ is delayed
     by $\delta_1$.
Since no further transition of~$s_{OR}$ occurs before time~$t_k+A+B$, and since
     $B > \delta_\infty-\delta_{\inf}$, it follows
     that~$s_{OR}$'s transition to~$0$ is not canceled by any
     following transition.
The transition of~$s_{OR}$ to~$1$ at time~$t_k+A+u$ is delayed
     by~$\delta_2$, where
$\delta_2 = \delta(u-\delta_1) \ge \delta(B-\delta_\infty)$, 
since $u\ge B$, $\delta_1\le \delta_\infty$ and $\delta$ is
     non-decreasing. Thus, by~(iv),
\begin{align}
\delta_2 \ge \delta_\infty - \varepsilon'\enspace.\label{eq:B2}
\end{align}

The transition of~$s_{OR}$ back to value~$0$ at time~$t_k+A+u+x$ is delayed
     by~$\delta_3$, where
\begin{align}
\delta_3 &= \delta(x-\delta_2) \leq \delta(C+\varepsilon'-\delta_\infty)\enspace,
\end{align}
since $x \leq C$, $\delta$ is non-decreasing, and by~\eqref{eq:B2}.
By~(iii),
\begin{align}
\delta_3 \leq \delta_{\inf} + \varepsilon\enspace.\label{eq:C2}
\end{align}
The pulse occurring during phase~$C$ is filtered
     out at the output~$c^*(s_{OR})$ of channel~$c^*$, since
     $\delta_2 \ge x + \delta_3$:
The latter follows from~\eqref{eq:B2}, (ii) and~\eqref{eq:C2}, as
$\delta_2 \ge \delta_\infty - \varepsilon' \ge \delta_{\inf} +
\varepsilon + C \ge \delta_3$. 

The transition of~$s_{OR}$ to value~$1$ at time~$t_{k+1} = t_k+A+u+x+y$ is
     delayed by~$\delta_4$, where 
$\delta_4 = \delta(y-\delta_3) \leq \delta(C+D-\delta_{\inf})$, 
since  $\delta$ is non-decreasing 
and $y \leq C+D$, $\delta(t) \ge \delta_{\inf}$ for all $t >
     -\delta_\infty$ such that $\delta_3=\delta(x-\delta_2)\geq\delta_{\inf}$.
By Assumption~(i), we may thus deduce
$\delta_4 \leq \delta_\infty - \Delta$. 
Since no further transition of $s_{OR}$ occurs before time~$t_{k+1}+A$, and $A
     > \delta_\infty - \delta_{\inf}$ by Assumption~(iv),
     $c^*(s_{OR})$'s transition at time~$t_{k+1}+\delta_4$ is not
     canceled by any later transition.
Since $A > \delta_\infty - \delta_{\inf} >
     E+F-\delta_{\inf}$, by Assumptions~(iv) and~(v), and the
     fact that a transition is delayed by at least
     time~$\delta_{\inf}$, no other transition of~$c^*(s_{OR})$
     occurs during $(t_{k+1}+\delta_4,t_{k+1}+E+F]$.
It follows that $c^*(s_{OR})$'s value is~$1$ during phase~$F$ of
     round~$k+1$.
\end{proof}

\begin{figure}[ht]
\centering
\begin{tikzpicture}[scale=1,>=latex']
  \path [draw,->] (0,0)--(0,{\highV+0.2}) node (xaxis) [above,xshift=-7mm,yshift=-5mm] {$s_{OR}(t)$};
  \path [draw,->] (0,0)--(7.2,0)   node (yaxis) [above] {$t$};

  \feline{0}{1.5}{\highV}
  \feline{1.5}{3.3}{0}
  \feline{3.3}{4.0}{\highV}
  \feline{4.0}{5.5}{0}
  \fnline{5.5}{7}{\highV}

  \draw (0,0) node[below] {$t_k$};

  \draw[dashed] (1.5,-0.4) -- ++(0,{\highV+0.6});
  \node[below] at (0.8,0) {$A$};

  \draw[dashed] (3.0,-0.4) -- ++(0,{\highV+0.6});
  \node[below] at (2.4,0) {$B$};

  \draw[dashed] (4.2,-0.4) -- ++(0,{\highV+0.6});
  \node[below] at (3.6,0) {$C$};

  \draw[dashed] (5.5,-0.4) -- ++(0,{\highV+0.6});
  \node[below] at (5,0) {$D$};

  \node[below] at (6.4,0) {$A$};

  \draw[<->] (1.5,0.15) -- (3.3,0.15) node[midway,above] {$u$};
  \draw[<->] (3.3,0.15) -- (4.0,0.15) node[midway,above] {$x$};
  \draw[<->] (4.0,0.15) -- (5.5,0.15) node[midway,above] {$y$};

  \path [draw,->] (0,\offV)--(0,{\offV+\highV+0.2}) node (xaxis) [above,xshift=-15pt] {$c^*(s_{OR})(t)$};
  \path [draw,->] (0,\offV)--(7.2,\offV) node (yaxis) [above] {$t$};

  \feline{0}{0.8}{\offV}
  \feline{0.8}{2.3}{{\offV+\highV}} 
  \feline{2.3}{5.7}{\offV} 
  \fnline{5.7}{7.0}{{\offV+\highV}}

  \draw (0,\offV) node[below] {$t_k$};

  \draw[dashed] (5.8,{\offV-0.4}) -- ++(0,{\highV+0.6});
  \node[below] at (3.0,\offV) {$A+B+C+D+E$};

  \draw[dashed] (6.2,{\offV-0.4}) -- ++(0,{\highV+0.6});
  \node[below] at (6.0,\offV) {$F$};

  \draw[dashed,->] (0,-0.5) -- (0.8,-1.1)
    node[pos=0.2,right,xshift=5pt] {$\delta_0$};
  \draw[dashed,->] (1.5,-0.5) -- (2.3,-1.1)
    node[midway,right] {$\delta_1$};
  \draw[dashed,->] (3.3,-0.5) -- (4.1,-1.1)
    node[midway,left] {$\delta_2$};
  \draw[dashed,->] (4.0,-0.5) -- (4.1,-1.1)
    node[midway,right] {$\delta_3$};
  \draw[dashed,->] (5.5,-0.5) -- (5.7,-1.1)
    node[midway,right] {$\delta_4$};

\end{tikzpicture}
\caption{Case 1: Input and Output of channel~$c^*$ in circuit~$C_\mathrm{NF}$ if
     phase~$C$ contains a pulse.}\label{fig:non-const:1} 
\end{figure}
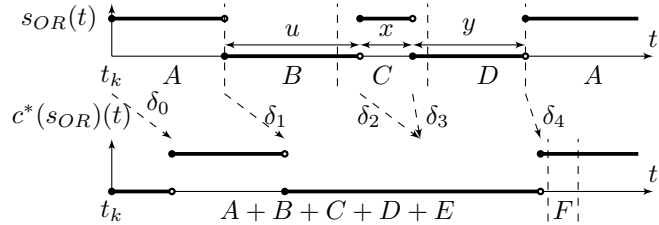

\begin{figure*}[t]
\centering
\hfill
\begin{minipage}[b]{0.57\textwidth}
\begin{tikzpicture}[circuit logic US, scale=0.8, large circuit symbols,every node/.style={transform shape}]

\useasboundingbox (-6.0,-1.4) rectangle (5.2,1.3);

\tikzstyle{osc} = [draw,rectangle,minimum width=1.1cm,minimum height=0.7cm]

\matrix[column sep=12mm,row sep=-1mm]
{
\node[osc,yshift=1.46mm] (o2) {$CLK_A$};    &                          & \node[or gate,logic gate inputs=nn] (or) {OR}; && \\
    &\node [and gate,logic gate inputs=nn] (and2) {AND};&                       & \node [and gate,logic gate inputs=nn] (and) {AND}; &  \\
\node[osc,yshift=1mm] (o3) {$CLK_C$}; && \node[osc,yshift=1mm] (o1)
     {$CLK_F$};                     && \\
};

\draw (or.output) -- ++(1.0,0) node [midway,above] {$c^*$} node [midway,below] {} |- (and.input 1);

\draw (and.output)
	-- ++(0.85,0) node [midway,above] {$\delta=1$} node[midway,below] {$x=0$} 
	node[right] {$o$};

\path[draw] (o2) -- (or.input 1) node [midway,above] {$\delta=2$} node[midway,below] {$x=0$}
                 ;

\path[draw] (o1) -- ++(1.85,0) node[midway,above] {$\delta=2$} node[midway,below] {$x=0$}
                 -- ++(0,0.1)
                 |- (and.input 2)
                 ;

\path[draw] (and2.input 1) -- ++(-1.8,0) node[midway,above] {$\delta=1$} node[midway,below] {$x=0$}
                 node[left] {$i$};
                 ;

\path[draw] (and2.input 2) -- ++(-0.2,0) 
                           |- (o3.east) node[pos=0.7,above] {$\delta=1$} node[pos=0.7,below] {$x=0$}
                           ;

\draw (and2.output) -- ++(1.2,0) node [midway,above] {$\delta=1$} node [midway,below] {$x=0$} |- (or.input 2);

\end{tikzpicture}
\caption{Circuit $C_\mathrm{NF}$ used in Cases~1 and~2.1.}\label{fig:nf}
\end{minipage}
\begin{minipage}[b]{0.42\textwidth}
\begin{tikzpicture}[circuit logic US, scale=0.8, large circuit symbols,every node/.style={transform shape}]

\useasboundingbox (-4.0,-1.6) rectangle (4.6,1.1);

\matrix[column sep=12mm, row sep=-1mm]
{
\node[buffer gate,small circuit symbols, logic gate inputs=n] (buff) {};
& \node[and gate,logic gate inputs=nn] (andx) {AND}; 
& \node [or gate,logic gate inputs=nn,yshift=3mm] (or) {OR};\\
};

\draw (andx.output) -- ++(1.15,0) node [midway,above] {$\delta=1$} node [midway,below] {$x=0$} |- (or.input 1);

\draw (or.output) 
	-- ++(right:0.1) 
	-- ++(down:1.0)
	-- ++(left:1.5) node[midway, above] {$\delta=\varepsilon$} node[midway,below] {$x=0$}
	-- ++(up:0.4)
	|- (or.input 2)
	;

\draw (or.output)
        -- ++(right:0.1) node[circle,inner sep=1pt,fill=black,draw] {}
        -- ++(right:1.0) node [midway,above] {$\delta=1$} node[midway,below] {$x=0$} 
        node[right] {$o$}
	;

\path[draw] (andx.input 1)
              -- ++(-0.2,0)
              -- ++(0,0.5)
              -- ++(-1.2,0) node[midway,above] {$\delta=1$} node[midway,below] {$x=0$}
              -- ++(0,-0.68)
              -- (buff.output) 
              ;

\path[draw] (andx.input 2)
              -- ++(-0.2,0)
              -- ++(0,-0.5)
              -- ++(-1.2,0) node[midway,above] {$\scriptstyle\delta=1+\varepsilon'$} node[midway,below] {$x=0$}
              -- ++(0,0.68) node[circle,inner sep=1pt,fill=black,draw] {}
              ;

\path[draw] (buff.input)-- ++(-0.7,0) node[midway,above] {$c^*$} node[left] {$i$};

\node at (-2.2,-1.5) {$\varepsilon'=\max(0,\delta^--\delta_{\inf})$};
\node at (1.8,-1.5) {$0 < \varepsilon < \delta_\infty-\delta_{\inf}-\varepsilon'$};

\end{tikzpicture}
\caption{Circuit $C_\mathrm{NC}$ used in Case~2.2.}\label{fig:nc}
\end{minipage}
\end{figure*}

\noindent {\bf Case~2.1.} In this case, we choose
\begin{enumerate}
\item[(i)] $A=D > \max(0,\delta_\infty-\delta_{\inf})$ and large
  enough such that $\delta(A-\delta_\infty) =
     \delta_\infty$.
Such an~$A$ must exist, because of the assumption of Case~2.1.

\item[(ii)] $B,C,\varepsilon > 0$ small enough such that
     $B+C+\varepsilon+\delta_{\inf} \leq \delta_\infty$.

\item[(iii)] $0 < \varepsilon' < B+C$

\item[(iv)] $\varepsilon > 0$ small enough such that
     $\delta(-\delta_{\inf}-\varepsilon) \ge \delta_\infty -
     \varepsilon'$.
Such a value exists, since $\delta$ is continuous at~$-\delta_{\inf}$
     by the assumption of Case~2.1.

\item[(v)] $B+C > 0$ small enough such that $\delta(B+C-\delta_\infty) \le
     \delta_{\inf}+\varepsilon$.

\item[(vi)] $E = A+\delta_\infty$ and $F = B+C-\varepsilon'$.

\end{enumerate}
Again, it is easy to verify that Assumptions~(i)-(vi) are
compatible with each other.

Figures~\ref{fig:non-const:0b} and~\ref{fig:non-const:1b} depict
     signal~$s_{OR}$ in absence and presence of a pulse.

\smallskip

We next show by induction on $k\ge 0$ that signal~$s_{OR}$'s
     transition at time~$t_k$ is delayed by~$\delta_\infty$, and that
     the channel's output~$c^*(s_{OR})$ has value~$0$ during
     phase~$F$ of round~$k$ in the absence of a pulse within phase~$C$
     of round~$k$, and value~$1$ in the presence of a pulse.

\smallskip


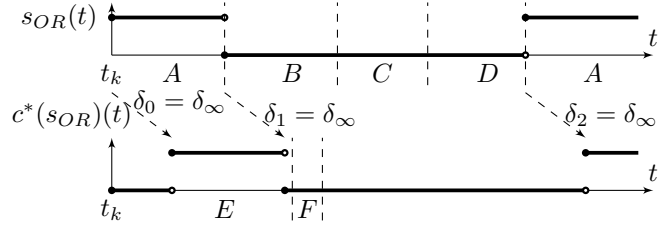
\begin{figure}[ht]
\centering
\begin{tikzpicture}[scale=1,>=latex']
  \path [draw,->] (0,0)--(0,{\highV+0.2}) node (xaxis) [above,xshift=-7mm,yshift=-5mm] {$s_{OR}(t)$};
  \path [draw,->] (0,0)--(7.2,0)   node (yaxis) [above] {$t$};

  \feline{0}{1.5}{\highV}
  \feline{1.5}{5.5}{0}
  \fnline{5.5}{7}{\highV}

  \draw (0,0) node[below] {$t_k$};

  \draw[dashed] (1.5,-0.4) -- ++(0,{\highV+0.6});
  \node[below] at (0.8,0) {$A$};

  \draw[dashed] (3.0,-0.4) -- ++(0,{\highV+0.6});
  \node[below] at (2.4,0) {$B$};

  \draw[dashed] (4.2,-0.4) -- ++(0,{\highV+0.6});
  \node[below] at (3.6,0) {$C$};

  \draw[dashed] (5.5,-0.4) -- ++(0,{\highV+0.6});
  \node[below] at (5,0) {$D$};

  \node[below] at (6.4,0) {$A$};

  \path [draw,->] (0,\offV)--(0,{\offV+\highV+0.2}) node (xaxis) [above,xshift=-15pt] {$c^*(s_{OR})(t)$};
  \path [draw,->] (0,\offV)--(7.2,\offV) node (yaxis) [above] {$t$};

  \feline{0}{0.8}{\offV}
  \feline{0.8}{2.3}{{\offV+\highV}} 
  \feline{2.3}{6.3}{\offV} 
  \fnline{6.3}{7.0}{{\offV+\highV}}

  \draw (0,\offV) node[below] {$t_k$};

  \draw[dashed] (2.4,{\offV-0.4}) -- ++(0,{\highV+0.6});
  \node[below] at (1.5,\offV) {$E$};

  \draw[dashed] (2.8,{\offV-0.4}) -- ++(0,{\highV+0.6});
  \node[below] at (2.6,\offV) {$F$};

  \draw[dashed,->] (0,-0.5) -- (0.8,-1.1)
    node[pos=0.2,right] {$\delta_0 = \delta_\infty$};
  \draw[dashed,->] (1.5,-0.5) -- (2.3,-1.1)
    node[midway,right] {$\delta_1 = \delta_\infty$};
  \draw[dashed,->] (5.5,-0.5) -- (6.3,-1.1)
    node[midway,right] {$\delta_2 = \delta_\infty$};

\end{tikzpicture}
\caption{Case 2.1: Input and Output of channel~$c^*$ in circuit~$C_\mathrm{NF}$
     if phase~$C$ does not contain a pulse.}\label{fig:non-const:0b} 
\end{figure}

\begin{proof}
Assume the input signal~$s_{OR}$ of channel~$c^*$ contains no pulse
     within phase~$C$ of round~$k$.
The signal is depicted in \figref{fig:non-const:0b}.

Signal~$s_{OR}$'s transition to value~$1$ at time~$t_k$ is delayed by
     some~$\delta_0$.
Clearly, if $k=0$ (i.e., in round~$0$), $\delta_0 = \delta_\infty$.
As induction hypothesis assume in the following that signal~$s_{OR}$'s
     transition at time~$t_{k}$ is delayed by $\delta_\infty$.
We will show that this implies that signal~$s_{OR}$'s transition at
     time~$t_{k+1}$ is delayed by~$\delta_\infty$.

Obviously, the next transition of~$s_{OR}$ back to value~$0$ at
     time~$t_k+A$ is delayed by $\delta_1$, where    
\begin{align}
\delta_1 = \delta(A-\delta_0) = \delta(A-\delta_\infty) = \delta_\infty\enspace,\label{eq:A1_0new}
\end{align}
by the choice of~$A$ according
     to Assumption~(i).
Further, by Assumption~(i), $A >
     \delta_\infty-\delta_{\inf}$, implying that no
     transition of~$s_{OR}$ after time~$t_k$ can cancel the transition
     of~$c^*(s_{OR})$ to~$1$ at time~$t_k+\delta_0$.

The transition of~$s_{OR}$ to value~$1$ at time~$t_{k+1} =
     t_k+A+B+C+D$ is delayed by $\delta_2$, where 
\begin{align*}
\delta_2 = \delta(B+C+D-\delta_1) = \delta(B+C+D-\delta_\infty) =
     \delta_\infty\enspace,
\end{align*}
because of Assumption~(i).
Thus, the initial transition of round~$k+1$ at time~$t_{k+1}$ will be
     delayed by $\delta_\infty$, which completes the inductive step.
Since $D > \delta_\infty-\delta_{\inf}>0$, by Assumption~(i),
     it follows that $c^*(s_{OR})$'s transition to~$0$ at
     time~$t_k+A+\delta_1$ is not canceled by any transition.
By analogous arguments, the transition to~$1$ at time
     $t_{k+1}+\delta_2$ is not canceled by any transition.
Our choice of~$E$ and~$F$ in (vi) thus implies that the channel
     output's value is~$0$ during phase~$F$ of round~$k$,
     see~\figref{fig:non-const:0b}.
\end{proof}



\begin{proof}
Now assume that there is a pulse within phase~$C$ of round~$k$.
The channel's input and output signals are depicted in
     \figref{fig:non-const:1b}.

Signal~$s_{OR}$'s initial transition to value~$1$ at time~$t_k$
     clearly is delayed by~$\delta_0 = \delta_\infty$ if~$k=0$.
As induction hypothesis assume in the following that $s_{OR}$'s
     transition at time~$t_k$ is delayed by $\delta_\infty$.
We will show that this implies that $s_{OR}$'s transition at
     time~$t_{k+1}$ is delayed by $\delta_\infty$.

By the same reasoning as in the proof before, $c^*(s_{OR})$'s
     transition to~$1$ at time~$t_k+\delta_0$ is not canceled by any
     following transition.
Further, $s_{OR}$'s transition back to value~$0$ at time~$t_k+A$ is
     delayed by $\delta_1 = \delta_\infty$.

The transition of~$s_{OR}$ to value~$1$ at time~$t_k+A+u$ is delayed
     by~$\delta_2$, where
$\delta_2 = \delta(u-\delta_1) \leq \delta(B+C-\delta_\infty) \leq
     \delta_{\inf}+\varepsilon$, 
by Assumption~(v). From~(ii), we further obtain
$u + \delta_2 \leq B+C+\delta_{\inf}+\varepsilon \le
 \delta_\infty$. 
It follows that this output transition cancels the last output
     transition to~$0$.

The transition of~$s_{OR}$ back to value~$0$ at time~$t_k+A+u+x$ is delayed
     by~$\delta_3$, where
$\delta_3 = \delta(x-\delta_2) \ge \delta(-\delta_{\inf}-\varepsilon) \ge
  \delta_\infty - \varepsilon'$, 
holds because of Assumption~(iv).

The transition of~$s_{OR}$ to value~$1$ at time~$t_{k+1}$ is delayed
     by~$\delta_4$, where 
$\delta_4 = \delta(y-\delta_3) \ge \delta(D-\delta_{\infty}) =
     \delta_{\infty}$, 
by Assumption~(i), which  completes the inductive step.

Moreover, since $D > \delta_\infty-\delta_{\inf}>0$, it follows
     that $c^*(s_{OR})$'s transition to~$0$ at
     time~$t_k+A+u+x+\delta_3$ is not canceled by any transition.
By similar arguments, $c^*(s_{OR})$'s transition to~$1$ at
     time~$t_{k+1}+\delta_4$ is not canceled by any following
     transition.

Assumption~(vi) hence implies  that $c^*(s_{OR})$'s value is~$1$
     during phase~$F$ of round~$k+1$, see~\figref{fig:non-const:1b}.
\end{proof}

\begin{figure}[ht]
\centering
\begin{tikzpicture}[scale=1,>=latex']
  \path [draw,->] (0,0)--(0,{\highV+0.2}) node (xaxis) [above,xshift=-7mm,yshift=-5mm] {$s_{OR}(t)$};
  \path [draw,->] (0,0)--(7.2,0)   node (yaxis) [above] {$t$};

  \feline{0}{1.5}{\highV}
  \feline{1.5}{3.3}{0}
  \feline{3.3}{4.0}{\highV}
  \feline{4.0}{5.5}{0}
  \fnline{5.5}{7}{\highV}

  \draw (0,0) node[below] {$t_k$};

  \draw[dashed] (1.5,-0.4) -- ++(0,{\highV+0.6});
  \node[below] at (0.8,0) {$A$};

  \draw[dashed] (3.0,-0.4) -- ++(0,{\highV+0.6});
  \node[below] at (2.4,0) {$B$};

  \draw[dashed] (4.2,-0.4) -- ++(0,{\highV+0.6});
  \node[below] at (3.6,0) {$C$};

  \draw[dashed] (5.5,-0.4) -- ++(0,{\highV+0.6});
  \node[below] at (5,0) {$D$};

  \node[below] at (6.4,0) {$A$};

  \draw[<->] (1.5,0.15) -- (3.3,0.15) node[midway,above] {$u$};
  \draw[<->] (3.3,0.15) -- (4.0,0.15) node[midway,above] {$x$};
  \draw[<->] (4.0,0.15) -- (5.5,0.15) node[midway,above] {$y$};

  \path [draw,->] (0,\offV)--(0,{\offV+\highV+0.2}) node (xaxis) [above,xshift=-15pt] {$c^*(s_{OR})(t)$};
  \path [draw,->] (0,\offV)--(7.2,\offV) node (yaxis) [above] {$t$};

  \feline{0}{0.8}{\offV}
  \feline{0.8}{4.5}{{\offV+\highV}} 
  \feline{4.5}{6.3}{\offV} 
  \fnline{6.3}{7.0}{{\offV+\highV}}

  \draw (0,\offV) node[below] {$t_k$};

  \draw[dashed] (2.4,{\offV-0.4}) -- ++(0,{\highV+0.6});
  \node[below] at (1.5,\offV) {$E$};

  \draw[dashed] (2.8,{\offV-0.4}) -- ++(0,{\highV+0.6});
  \node[below] at (2.6,\offV) {$F$};

  \draw[dashed,->] (0,-0.5) -- (0.8,-1.1)
    node[pos=0.2,right] {$\delta_0 = \delta_\infty = \delta_1$}; 
  \draw[dashed,->] (1.5,-0.5) -- (2.3,-1.1);
  \draw[dashed,->] (3.3,-0.5) -- (2.3,-1.1)
    node[midway,right] {$\delta_2$};
  \draw[dashed,->] (4.0,-0.5) -- (4.5,-1.1)
    node[midway,right] {$\delta_3$};
  \draw[dashed,->] (5.5,-0.5) -- (6.3,-1.1)
    node[midway,right] {$\delta_4=\delta_\infty$};

\end{tikzpicture}
\caption{Case 2.1: Input and Output of channel~$c^*$ in circuit~$C_\mathrm{NF}$ if
     phase~$C$ contains a pulse.}\label{fig:non-const:1b} 
\end{figure}
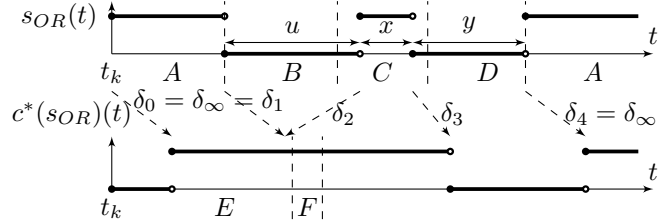

\noindent {\bf Case~2.2.} For this case, circuit~$C_\mathrm{NC}$ depicted in
     \figref{fig:nc} solves bounded SPF.
The algorithm and its proof rest on the following idea: We 
   first show in Lemma~\ref{lem:pulse_window}
that every channel~$c^*$ whose $\delta$ is in accordance with Case~2.2
does not produce pulses of length
within the non-zero interval~$[\max(0,\delta^--\delta_{\inf}),\delta_\infty-\delta_{\inf})$.
The remaining part of circuit~$C_\mathrm{NC}$ thus just has to
filter out all pulses with duration 
     less than~$\max(0,\delta^--\delta_{\inf})$ (ensured by the 
AND gate) and continuously hold
     all pulses of length~$\delta_\infty-\delta_{\inf}$ (done by
the OR gate).

\begin{lem}\label{lem:pulse_window}
Let $c^*$ be a non-constant-delay non-forgetful channel chosen in
     accordance to Case~2.2.
If the channel's input signal is a pulse, then its output signal is
     either~$0$ or a pulse whose length is not
     within the non-zero interval~$[\max(0,\delta^--\delta_{\inf}),\delta_\infty-\delta_{\inf}]$. 
\end{lem}
\begin{proof}
Assume that $\delta(-\delta_{\inf}) = \delta_\infty$; the
     proof for the case $\delta(-\delta_{\inf})=\delta^- < \delta_\infty$ is
     almost the same. 
Without loss of generality, assume that the input pulse
starts at
     time~$0$ and let~$x>0$ be its length.
Clearly, the transition of the output signal to~$1$ is scheduled at
     time~$\delta_\infty$, the transition back to~$0$ is scheduled at
     time~$x+\delta(x-\delta_\infty)$.
We distinguish two cases for the input pulse length~$x$: 

\smallskip

In case $x < \delta_\infty-\delta_{\inf}$, we 
have $\delta(x-\delta_\infty)\leq\delta^-$ and the following two sub-cases:
If additionally $x \leq \delta_\infty-\delta^-$, then
$x+\delta(x-\delta_\infty) \leq x+\delta^- \leq \delta_\infty$, 
so the output events cancel. If $\delta_\infty-\delta_{\inf} > x > \delta_\infty-\delta^-$,
the length of the output pulse is $x+\delta(x-\delta_\infty) - \delta_\infty < \delta^- - 
\delta_{\inf}$. This confirms the lower boundary of the ``forbidden pulse length interval'' 
given in our lemma.
In case of $x \ge \delta_\infty-\delta_{\inf}$, on the other hand,
$\delta(x-\delta_\infty) = \delta_\infty$ 
a pulse with length
$x+\delta(x-\delta_\infty)-\delta_\infty \ge
\delta_\infty-\delta_{\inf}$ 
is generated at the output of $c^*$, which also confirms the 
upper boundary of the interval.
\end{proof}

If we choose the circuit parameters in \figref{fig:nc} according to
$\varepsilon' =
     \max(0,\delta^--\delta_{\inf})$ and $0 < \varepsilon <
     \delta_\infty-\delta_{\inf}-\varepsilon'$, it is not difficult
to show that the resulting circuit $C_\mathrm{NC}$ solves bounded SPF 
in Case~2.2: Properties~(F1) to~(F3) trivially hold for
     circuit~$C_\mathrm{NC}$. To prove (F4), consider that
if the input signal~$i$ is a pulse of
     length~$2\delta_\infty$, the output signal~$s_{c^*(i)}$ of $c^*$ is a
     pulse of length at least~$\delta_\infty$.
Thus, the output of the AND gate $s_{AND}$ is a pulse of length 
at least~$\delta_\infty -
     \varepsilon' > \varepsilon$, resulting in the circuit's output
     $o$ making a transition to~$1$ and remaining~$1$ from there on.

Property~(F5) directly follows from Lemma~\ref{lem:pulse_window}: If
     $s_{c^*(i)}$ is a pulse of length smaller
     than~$\max(0,\delta^--\delta_{\inf}) = \varepsilon'$, then it is
     completely filtered out; $s_{AND}$ and hence $o$ are hence permanently~$0$.
Otherwise, by Lemma~\ref{lem:pulse_window}, $s_{c^*(i)}$ must be a pulse of
     length at least~$\delta_\infty-\delta_{\inf}$.
Thus, $s_{AND}$ is a pulse of length at
     least~$\delta_\infty-\delta_{\inf}-\varepsilon' > \varepsilon$,
which is sufficiently long to be permanently captured in the storage
looped formed by the OR gate. The circuit's output $o$ hence 
makes a transition to~$1$ and remains~$1$ from there~on.

Finally, (F6) is due to bounded channel delays.

\section{Eventual Short-Pulse Filtration with Constant Delays}\label{sec:eSPF}

We proved that SPF is not solvable with
     constant-delay channels.
In this section, we consider the weaker eventual SPF problem, which
     drops the ``no short pulses'' requirement (F5) and replaces it
     with its eventual analogon (F5e).
We show that eventual SPF is solvable using only constant-delay channels.
More specifically, we prove that
     circuit~$C_\mathrm{ev}$ in \figref{fig:Cev} solves eventual SPF.
The circuit contains a delay parameter~$\alpha$, which we will choose to be a
positive irrational like $\alpha=\sqrt{2}$.

We will show that the circuit's output is eventually stable at~$1$ whenever the input is
a pulse of positive length.
We derive a bound on this stabilization time
     in terms of the input pulse length~$\Delta$.
The bound is almost linear in~$1/\Delta$: It is in the order of~$O(\Delta^{-1 -
\varepsilon})$ for all~$\varepsilon>0$.

\begin{figure}[tb]
\centering
\begin{tikzpicture}[circuit logic US, scale=1.5, large circuit symbols]
\matrix[column sep=12mm]
{
\node (i) {$i$}; & & \node [or gate,inputs={nnn}] (or) {OR}; &  \\
};

\draw (i) -- ++(1.4,0) node [pos=0.2,above] {$\delta=1$} node [pos=0.2,below] {$x=0$} |- (or.input 2);
\draw (or.output) 
	-- ++(right:0.1)node[circle,inner sep=1pt,fill=black,draw] {}
	-- ++(down:0.4)
	-- ++(left:1.6) node[pos=0.7, above] {$\delta=1$} node[pos=0.7,below] {$x=0$}
	-- ++(up:0.1)
	|- (or.input 3)
	;
\draw (or.output) 
	-- ++(right:0.1)node[circle,inner sep=1pt,fill=black,draw] {}
	-- ++(up:0.4)
	-- ++(left:1.6) node[pos=0.7, above] {$\delta=\alpha$} node[pos=0.7,below] {$x=0$}
	-- ++(down:0.1)
	|- (or.input 1)
	;
\draw (or.output)
	-- ++(1.3,0) node [midway,above] {$\delta=1$} node[midway,below] {$x=0$} 
	node[right] {$o$}
	;
\end{tikzpicture}
\caption{Circuit $C_\mathrm{ev}$ solving eventual SPF.}
\label{fig:Cev}
\end{figure}
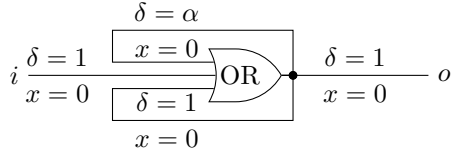

The measure
     points of circuit~$C_\mathrm{ev}$ for time~$t$ are of the form $t
     - (\alpha k + \ell) - 2$, where $k$ and $\ell$ are
     nonnegative integers. 
We can hence characterize the circuit's behavior with the following obvious lemma.

\begin{lem}\label{lem:ev:exec}
In every execution~$(s_v)$ of circuit~$C_\mathrm{ev}$, the following are equivalent:
(i) $s_o(t)=1$, and (ii) there exist nonnegative integers~$k$ and~$\ell$ such that
$ s_i\big( t - (\alpha k + \ell) -2 \big)=1$. 
\end{lem}

We may restrict our considerations to input pulses starting at time~$0$.
In the following, let the input signal~$s_i$ be a pulse of length~$\Delta>0$.
We are looking for the {\em stabilization time}, which is the minimal time $T=T(\Delta)$ such
     that, for all $t\geq T$, we have $s_o(t)=1$.

To prove finiteness and effective bounds on the stabilization time, we relate
it to the number-theoretic concept of {\em discrepancy\/} of the sequence
$(\alpha n)$ modulo~$1$ (see, e.g., \cite{DT97}).
The discrepancy compares the number of sequence elements in a given interval
with their expected number if the elements were uniformly distributed.

For a given nonempty subinterval $(x,y]$ of $(0,1]$ and a given positive integer~$N$, let
$A(x,y;N)$ denote the number of $\alpha n$'s with $n\leq N$ that lie in the
interval modulo~$1$: 
$\alpha n\in (x,y] + \mathds{Z}$.
The expected number of such $\alpha n$'s is $(y-x)N$.
The discrepancy $D_N(\alpha)$ is then defined as the maximum difference
between $A(x,y;N)$ and $(y-x)N$, formed over all nonempty subintervals $(x,y]$ of
$(0,1]$.

It is well-known that $D_N(\alpha) / N\to 0$ if and only if $\alpha$
     is irrational.
Also, if~$\alpha$ has a bounded continued fraction expansion, then
     $D_N(\alpha) = O(\log N)$ and the constant can be computed~\cite{schoi93}.
This is, in particular, true for~$\alpha=\sqrt{2}$.

\begin{lem}\label{lem:k:and:t}
Let~$K=K(\Delta)$ be the least integer~$K$ such that for all
     real~$t$ there exists an integer $k$, $0\leq k\leq K$,
     with $\alpha k\in(t-\Delta,t]+\mathds{Z}$.
Then, 
$T(\Delta)
     \leq \alpha\cdot K(\Delta) + \Delta + 2$. 
\end{lem}
\begin{proof}
The lemma is trivial if~$K=\infty$, so assume the contrary.

Let $t\geq\alpha K + \Delta +2$.
By the definition of~$K$, there exists a~$k$ with $0\leq k\leq K$ and an
     $\ell$ such that $ t - \Delta - \ell -2 < \alpha k
     \leq t - \ell -2 $, which is equivalent to $0 \leq t -
     (\alpha k + \ell) -2 < \Delta$.

By Lemma~\ref{lem:ev:exec}, it remains to prove that $\ell$ is nonnegative.
The inequality  $t - (\alpha k + \ell) -2 < \Delta$ is equivalent to
     $\ell>t-\Delta-\alpha k - 2$.
Noting $-\alpha k \geq - \alpha K$ and $t\geq\alpha K + \Delta+2$ shows
     $\ell>0$ and concludes the proof.
\end{proof}




\begin{lem}\label{lem:d:and:k}
Let~$0<\Delta \leq1$.
If $D_N(\alpha) / N < \Delta/2$, then $K(\Delta)\leq N$.
\end{lem} 
\begin{proof} Suppose the contrary, i.e., that there exists a real~$t$
such that, for all $n\leq N$, we have $\alpha n \not\in
     (t-\Delta,t]+\mathds{Z}$.
Let $0<x<y\leq z< u\leq 1$ such that we can decompose the interval 
$(t - \Delta, t] +
     \mathds{Z} = \big( (x,y] + \mathds{Z} \big) \cup \big( (z,u] + \mathds{Z} \big)$
modulo~$1$.
None of the two intervals~$(x,y]$ and~$(z,u]$ contains an $\alpha n$ modulo~$1$
with~$n\leq N$.
Hence $A(x,y;N)=A(u,z;N)=0$, which implies $2 D_N(\alpha) \geq (y - x)N
     + (u - z)N = \Delta N$, a contradiction.
\end{proof}

\begin{thm}\label{thm:ev}
Circuit~$C_\mathrm{ev}$ solves eventual SPF if~$\alpha$ is irrational.
If~$\alpha=\sqrt{2}$, the stabilization time satisfies
$T(\Delta) = O(\Delta^{ -1-\varepsilon})$
as~$\Delta\to0$
for all~$\varepsilon >0$.
\end{thm}
\begin{proof}
(F1) and (F2) are obviously fulfilled.
Because all initial values of channels are~$0$, also (F3) holds.
Because $D_N(\alpha)/N\to0$ whenever~$\alpha$ is irrational,
for all~$\Delta>0$, there exists some~$N$ such that $D_N(\alpha)/N < \Delta/2$.
Hence Lemma~\ref{lem:d:and:k} and Lemma~\ref{lem:k:and:t} show that~$T(\Delta)$
is finite, which shows (F4) and (F5e).

We now prove the bound on the stabilization time.
Let $\gamma = - 1 - \varepsilon < -1$.
There exists a~$C_1>0$ such that $D_N(\alpha) \leq C_1 \log N$.
Because $ 1+ 1/\gamma > 0$, there exists a~$C_2>0$ such that $\log N <
C_2  N^{1 + 1/\gamma}$.
Thus if 
\[N\geq \left (\frac{\Delta}{2C_1C_2}\right)^{  \gamma }\]
then
\[ \frac{D_N(\alpha)}{N} \leq \frac{C_1\log N}{N} < C_1C_2 N^{ 1/\gamma } \leq \frac{\Delta}{2} \enspace,  \]
which, by Lemma~\ref{lem:d:and:k}, implies
\begin{equation*}
K(\Delta) \leq  \left( \frac{\Delta}{2C_1C_2} \right)^{\gamma } + 1 
\end{equation*}
for all~$0<\Delta\leq1$.
That is, $K(\Delta) = O(\Delta^{\gamma})$ as~$\Delta\to0$.

It is easy to see that $K(\Delta)\to\infty$ as $\Delta\to0$.
Hence Lemma~\ref{lem:k:and:t} implies~$T(\Delta) = O( K(\Delta) )$ as~$\Delta\to 0$ as
asserted.
\end{proof}

\section{Conclusion}\label{sec:conc}

We showed that binary circuit models using bounded single-history
     channels, hence all binary models known to date, fail to 
faithfully model glitch propagation: In case of constant-delay
channels, SPF turned out to be unsolvable, which is
     in contradiction to physical reality.
In case of non-constant-delay channels, even bounded SPF is
     solvable, again in contradiction to physical reality.
Future binary models aiming at faithful glitch propagation modeling hence
     cannot have the bounded single-history property.

We hope that our results provide a signpost for future research on
     adequate binary circuit models:
As confirmed by the fact that the weaker eventual
SPF problem is already solvable with constant-delay channels, 
SPF is well suited for capturing the peculiarities of glitch 
propagation while not being overly restrictive. Moreover,
in the proofs of our core results,
we actually used weaker properties
than actually guaranteed by single-history channels. It may
hence be possible to re-use part of those for alternative weaker
channel models.

%
%

%
%

\end{document}